\documentclass{article}

\usepackage{fullpage}
\usepackage{lmodern}
\usepackage{authblk}
\usepackage{amsmath, amssymb, mathtools}
\allowdisplaybreaks
\usepackage[ruled]{algorithm2e}
\usepackage{algorithmic}
\usepackage{enumitem}
\usepackage{tikz}
\usepackage{hyperref}
\usepackage[small]{caption}
\usepackage[numbers]{natbib}

\newtheorem{theorem}{Theorem}[section]

\newtheorem{lemma}{Lemma}[section]

\newtheorem{definition}{Definition}
\newtheorem{observation}{Observation}

\newtheorem{assumption}{Assumption}
\newtheorem{remark}{Remark}
\newenvironment{proof}{\paragraph{Proof:}}{\hfill$\square$}

\title{Reallocating Multiple Facilities on the 
Line\thanks{
		Loukas Kavouras is partially supported by a scholarship from the State 
		Scholarships Foundation, granted by the action ``Scholarships Grant 
		Programme for second cycle graduate studies", which is co-financed by  
		Greece and the European Union (European Social Fund- ESF) through the 
		Operational Programme "Human Resources Development, Education and 
		Lifelong Learning 2014- 2020". Stratis Skoulakis is partially supported by a 
		scholarship of Onnasis Foundation. Philip Lazos is supported by the ERC 
		Advanced Grant 788893 (AMDROMA). The names of the authors are in 
		alphabetical order.}
	}
\author[1]{Dimitris Fotakis}
\author[1]{Loukas Kavouras}
\author[1]{Panagiotis Kostopanagiotis}
\author[2]{Philip Lazos}
\author[1]{Stratis Skoulakis}
\author[1]{Nikolas Zarifis}

\affil[1]{National Technical University of Athens}
\affil[2]{Sapienza University of Rome}
\date{}
\begin{document}

\maketitle

\begin{abstract}
 We study the multistage $K$-\emph{facility reallocation problem} on the real 
 line, where we maintain $K$ facility locations over $T$ stages, based on the 
 stage-dependent locations of $n$ agents. Each agent is connected to the 
 nearest facility at each stage, and the facilities may move from one stage to 
 another, to accommodate different agent locations. The objective is to minimize 
 the connection cost of the agents plus the total moving cost of the facilities, 
 over all stages. $K$-\emph{facility reallocation} was introduced 
 by~\Citet{KW2018}, where they mostly focused on the special case of a single 
 facility. Using an LP-based approach, we present a polynomial time algorithm 
 that computes the optimal solution for any number of facilities. We also 
 consider online $K$-\emph{facility reallocation}, where the algorithm becomes 
 aware of agent locations in a stage-by-stage fashion. By exploiting an 
 interesting connection to the classical $K$-\emph{server problem}, we present a 
 constant-competitive algorithm for $K = 2$ facilities.
\end{abstract}

\section{Introduction}

Facility Location is a classical problem that has been widely studied in both 
combinatorial optimization and operations research, due to its many practical 
applications. It provides a simple and natural model for industrial planning, 
network design, machine learning, data clustering and computer vision 
\citet{DH2002,L2011,CGM16,BSU13}. In its basic form, $K$-Facility Location 
instances are defined by the locations of $n$ agents in a metric space. The goal 
is to find $K$ facility locations so as to minimize the sum of distances of the 
agents to their nearest facility.

In many natural location and network design settings, agent locations are not 
known in advance. Motivated by this fact, \citet{M2001} introduced online facility 
location problems, where agents arrive one-by-one and must be 
\emph{irrevocably} assigned to a facility upon arrival. Moreover, the fast 
increasing volume of available data and the requirement for responsive services 
has led to new, \emph{online} clustering algorithms~\citet{LSS2016},
balancing the quality of the clusters with their rate of change over time. In 
practical settings related to online data clustering, new data points arrive, and 
the decision of clustering some data points together should not be regarded as 
irrevocable (see e.g., \citet{F2011} and the references therein). 

More recently, understanding the dynamics of temporally evolving social or 
infrastructure networks has been the central question in many applied areas such 
as viral marketing, urban planning etc. \emph{Dynamic facility location} 
proposed by \citet{EMS2014} has been a new tool to analyze temporal aspects 
of such networks. In this time dependent variant of facility location, agents may 
change their location over time and we look for the best tradeoff between the 
optimal connections of agents to facilities and the stability of solutions between  
consecutive timesteps. The stability of the solutions is modeled by introducing 
an additional moving cost (or switching cost), which has a different definition 
depending on the particular setting.

\smallskip\noindent\textbf{Model and Motivation.}
In this work, we study the multistage $K$-\emph{facility reallocation problem} on 
the real line, introduced by \Citet{KW2018}. 
In \emph{$K$-facility reallocation}, $K$ facilities are initially located at 
$(x_1^0,\ldots,x_K^0)$ on the real line. Facilities are meant to serve $n$ agents 
for the next $T$ days. At each day, each agent connects to the facility closest to 
its location and incurs a connection cost equal to this distance. The locations of 
the agents may change every day, thus we have to move facilities accordingly in 
order to reduce the connection cost. Naturally, moving a facility is not for free, 
but comes with the price of the distance that the facility was moved. Our goal is 
to specify the exact positions of the facilities at each day so that the total 
connection cost plus the total moving cost is minimized over all $T$ days. In the 
online version of the problem, the positions of the agents at each stage $t$ are 
revealed only after determining the locations of
the facilities at stage $t-1$.

For a motivating example, consider a company willing to advertise its products. 
To this end, it organizes $K$ advertising campaigns at different locations of a 
large city for the next $T$ days. Based on planned events, weather forecasts, 
etc., the company estimates a population distribution over the locations of the 
city for each day. Then, the company decides to compute the best possible 
campaign reallocation with $K$ campaigns over all days (see also 
\citep{KW2018} 
for more examples).  

\Citet{KW2018} fully characterized the optimal offline and online algorithms for 
the 
special case of a single facility and presented a dynamic programming algorithm 
for $K \geq 1$ facilities with running time exponential in $K$. Despite the 
practical significance and the interesting theoretical properties of $K$-facility 
reallocation, its computational complexity and its competitive ratio (for the online 
variant) are hardly understood. 

\smallskip\noindent\textbf{Contribution.}
In this work, we resolve the computational complexity of $K$-\emph{facility 
reallocation} on the real line and take a first step towards a full understanding of 
the competitive ratio for the online variant.  
More specifically, in Section~\ref{s:offline}, we present an optimal algorithm with 
running time polynomial in the combinatorial parameters of $K$-\emph{facility 
reallocation} (i.e., $n$, $T$ and $K$). This substantially improves on the 
complexity of the algorithm, presented in \citep{KW2018}, that is exponential in 
$K$. Our algorithm solves a Linear Programming relaxation and then 
\emph{rounds} the \emph{fractional solution} to determine the positions of the 
facilities. The main technical contribution is showing that a simple rounding 
scheme yields an integral solution
that has the exact same cost as the \emph{fractional one}.

Our second main result concerns the \emph{online version} of the problem with 
$K = 2$ facilities. 
%
We start with the observation that online $K$-\emph{facility reallocation problem} 
with $K \geq 2$ facilities is a natural and interesting generalization of the 
classical $K$-\emph{server problem}, which has been a driving force in the 
development of online algorithms for decades. The key difference is that, in the 
$K$-\emph{server problem}, there is a single agent that changes her location at 
each stage and a single facility has to be relocated to this new location at each 
stage. Therefore, the total connection cost is by definition $0$, and we seek to 
minimize the total moving cost. 

From a technical viewpoint, the $K$-\emph{facility reallocation problem} poses a 
new challenge, since it is \emph{much} harder to track the movements of the 
optimal algorithm as the agents keep coming. It is not evident at all whether 
techniques and ideas from the $K$-\emph{server problem} can be applied to the 
$K$-\emph{facility reallocation} problem, especially for more general metric 
spaces. As a first step towards this direction, we design a constant-competitive 
algorithm, when $K=2$. Our algorithm appears in Section~\ref{s:online} and is 
inspired by the \emph{double coverage
algorithm} proposed for the $K$-\emph{server problem} \citet{K2009}.  

\smallskip\noindent\textbf{Related Work.}
We can cast the \emph{$K$-facility reallocation problem} as a clustering problem 
on a temporally evolving metric. From this point of view, \emph{$K$-facility 
reallocation problem} is a dynamic $K$-\emph{median} problem. A closely 
related problem is the \emph{dynamic facility location problem}, 
\citet{EMS2014,AFS2017}. Other examples in this setting are the \emph{dynamic 
sum radii clustering}~\citet{BS2017} and multi-stage optimization problems on 
matroids and graphs~\citet{GTW2014}. 

In \citet{FS2011}, a mobile facility location problem was introduced, which can 
be 
seen as a
one stage version of our problem. They showed that even this version of the 
problem is $NP$-hard in general metric spaces using an approximation 
preserving reduction to $K$-\emph{median problem}.

Online facility location problems and variants have been extensively studied in 
the literature, see \citet{F2011} for a survey. \citet{DI2011} studied an online 
model, 
where facilities can be moved with zero cost. As we have mentioned before, the 
online variant of the \emph{$K$-facility reallocation problem} is a generalization 
of the \emph{$K$-server problem}, which is one of the most natural online 
problems.
\citet{K2009}
showed a $(2K-1)$-competitive algorithm for the \emph{$K$-server problem} for 
every metric space, which is also
$K$-competitive, in case the metric is the real line \citet{BK2000}. Other variants 
of
the \emph{$K$-server problem} include the \emph{$(H,K)$-server problem}
\citet{BEJK2017,BEJKP2015}, the \emph{infinite server problem} \citet{CKL2017} 
and
the \emph{$K$-taxi problem}\citet{FRR1990,CK2018}.

\begin{figure}[ht]
\centering
$
\begin{array}{lr}
  (1) & \hspace{-2cm}\text{min}\sum _{t=1}^T \bigg[ \sum\limits_{i\in 
  C}\sum\limits_{j\in V} d(\text{Loc}(i,t),j)x_{ij}^t + \sum\limits_{k\in F} S_k^t \bigg] 
  \\\\

 \sum \limits_{j \in V}x_{ij}^t=1  &\forall i \in C,t \in \{1,T\}\\\\
	x_{ij}^t \leq c_j^t&\forall i \in C,j \in V,t \in \{1,T\}\\\\
     c_j^t=\sum\limits_{k\in F} f_{kj}^t & \forall j \in V, t\in \{1,T\}\\\\
    \sum \limits_{j \in V} f_{kj}^t=1 & \forall k\in F, t\in \{1,T\}\\\\
        S_k^t = \sum \limits_{j,l \in V} d(j,l)S_{kjl}^t & \forall k \in F, t\in \{1,T\} \\\\
        \sum \limits_{j \in V} S_{kjl}^t =f_{kl}^t & \forall k \in F,l \in V,t\in \{1,T\}\\\\
        \sum \limits_{l \in V} S_{kjl}^t = f_{kj}^{t-1} & \forall k \in F,j \in V,t\in \{1,T\}\\\\
        x_{ij}^t,f_{kj}^t,S_{klj}^t \in \{0,1 \} & \forall k \in F,j \in V,t\in \{1,T\}\\\\
\end{array}
$
\caption{Formulation of \emph{K-facility reallocation}}\label{LP}
\end{figure}

\section{Problem Definition and Preliminaries}
\label{sec:prelim}

\begin{definition}[\emph{$K$-Facility Reallocation Problem}]\label{d:problem}
We are given a tuple $(x^0, C)$ as input. The $K$ dimensional vector 
$x^0=(x^0_1,\ldots,x^0_K)$
describes the initial positions of the facilities. The positions of the agents over 
time
are described by $C = (C_1, \ldots, C_T)$. The position of agent $i$ at stage $t$ 
is
$\alpha_i^t$ and $C_t = (\alpha_1^t, \ldots , \alpha_n^t)$ describes the 
positions 
of the agents at stage $t$.
\end{definition}

\begin{definition}
A solution of K-Facility Reallocation Problem is a sequence $x = (x^1, \ldots 
,x^T)$.
Each $x^t= (x_1^t,\ldots,x_K^t)$ is a $K$ dimensional vector that gives
the positions of the facilities at stage $t$ and
$x_k^t$ is the position of facility $k$ at stage $t$. The cost of the solution $x$ is
\[ Cost(x) = \sum_{t=1}^{T} \bigg[\sum 
\limits_{k=1}^K|x_{k}^t-x_{k}^{t-1}|+\sum\limits_{i=1}^{n} \min_{1\leq k \leq 
K}|\alpha_i^t-x_k^t| \bigg]\]
\end{definition}
Given an instance $(x^0,C)$ of the problem, the goal is to find a solution $x$
that minimizes the $Cost(x)$. The term 
$\sum_{t=1}^{T}\sum_{k=1}^K|x_{k}^t-x_{k}^{t-1}|$ describes
the cost for moving the facilities from place to place and we refer to it as 
\emph{moving cost}, while 
the term $\sum_{t=1}^{T}\sum_{i=1}^{n} \min_{1\leq k \leq K}|\alpha_i^t-x_k^t|$ 
describes the connection
cost of the agents and we refer to it as \emph{connection cost}.

In the online setting, we study the special case of $2$-\emph{facility reallocation 
problem}. We evaluate the performance of our algorithm using competitive 
analysis; an algorithm is $c$-competitive if for every request sequence, its online 
performance is at most $c$ times worse (up to a small additive constant) than 
the optimal \emph{offline} algorithm, which knows the entire sequence in 
advance.

\section{Polynomial Time Algorithm}
\label{s:offline}

Our approach is a typical LP based algorithm that consists of two basic steps.

\begin{itemize}
 \item \textbf{Step 1:} Expressing the $K$-\emph{Facility Reallocation Problem} 
 as an Integer Linear Program.
 \item \textbf{Step 2:} Solving \emph{fractionally} the Integer Linear Program and 
 \emph{rounding} the fractional solution to an integral one.
\end{itemize}

\subsection{Formulating the Integer Linear Program}
A first difficulty in expressing the $K$-\emph{Facility Reallocation Problem} as 
an Integer Linear Program is that the positions on the real line are infinite. We 
remove this obstacle with help of the following lemma proved in \citep{KW2018}.

\begin{lemma}\label{l:discrete}
Let $(x_0,C)$ an instance of the $K$-facility reallocation problem.
There exists an optimal solution $x^*$ such that for all stages
$t \in \{1,T\}$ and $k \in \{1,K\}$,
\[x_k^{*t} \in C_1 \cup \ldots \cup C_T \cup x^0\]
\end{lemma}

According to Lemma~\ref{l:discrete}, there exists an optimal solution that locates 
the facilities only at positions where either an agent has appeared or a facility 
was initially lying.
Lemma~\ref{l:discrete} provides an exhaustive search algorithm for the problem 
and is also the basis
for the \emph{Dynamic Programming} approach in \citep{KW2018}. We use 
Lemma~\ref{l:discrete} to formulate our Integer Linear Program.

The set of positions $Pos=C_1 \cup \ldots \cup C_T \cup x^0$ can be 
represented equivalently by a path $P = (V,E)$.
In this path, the $j$-th node corresponds to the $j$-th leftmost position of $Pos$ 
and the distance between two consecutive nodes on the path equals the 
distance of the respective positions on the real line.
Now, the \emph{facility reallocation problem} takes the following 
\emph{discretized form}:
We have a path $P = (V,E)$ that is constructed by the specific instance $(x^0,C)$.
Each facility $k$ is initially located at a node $j \in V$ and at each stage $t$, 
each agent $i$ is also located at a node of $P$.
The goal is to move the facilities from node to node such that the connection 
cost of the agents plus the moving cost
of the facilities is minimized.

To formulate this discretized version as an Integer Linear Program,
we introduce some additional notation. Let $d(j,l)$ be the distance of the nodes 
$j,l \in V$ in $P$, $F$ be the set of facilities
and $C$ be the set of agents. For each $i \in C$, $\text{Loc}(i,t)$ is the node 
where agent $i$ is located at stage $t$.
We also define the following $\{0,1\}$-indicator variables for all $t\in \{1,T\}$:
$x_{ij}^t =1$ if, at stage $t$, agent $i$ connects to a facility located at node $j$,
$f_{kj}^t =1$ if, at stage $t$, facility $k$ is located at node $j$,
$S_{kjl}^t =1$ if facility $k$ was at node $j$ at stage $t-1$ and moved to node 
$l$ at stage $t$.
Now, the problem can be formulated as the Integer Linear Program depicted in 
Figure\ref{LP}.

The first three constraints correspond to the fact that at every stage $t$, each 
agent $i$ must be connected
to a node $j$ where at least one facility $k$ is located.
The constraint $\sum_{j \in V} f_{kj}^t=1$ enforces each facility $k$ to be located 
at exactly one node $j$.
The constraint $S_k^t = \sum_{j,l \in V} d(j,l)S_{kjl}^t$ describes the cost for 
moving facility $k$ from node $j$ to node $l$.
The final two constraints ensure that facility $k$ moved from node $j$ to node 
$l$ at stage $t$ if and only if
facility $k$ was at node $j$ at stage $t-1$ and was at node $l$ at stage $t$
($S_{kjl}^t = 1$ iff $f_{kl}^t=1$ and $f_{kj}^{t-1}=1$). 

We remark that the values of $f_{kj}^0$ are determined by the initial positions
of the facilities, which are given by the instance of the problem. The notation 
$x_{ij}^t$ should not be confused with
$x_k^t$, which is the position of facility $k$ at stage $t$ on the real line .

\subsection{Rounding the Fractional Solution}

\begin{algorithm}[t]
    \SetAlgoLined
    \DontPrintSemicolon
    \KwData{Given the initial positions $x^0 = \{x_1^0,\ldots,x_K^0\}$ of the 
    facilities and
    the positions of the agents $C = \{C_1,\ldots,C_T\}.$}
\begin{itemize}[leftmargin=*]
 \item Construct the path $P$ and the Integer Linear Program~(\ref{LP}).
 \item Solve the relaxation of the Integer Linear Program~(\ref{LP}).
 \item \emph{Rounding}
 \footnotemark: For each stage $t\geq 1$:
 \begin{itemize}
  \item For $m=1,\ldots,K$, find the node  $j_m^t$ such that 
    \[\sum_{\ell = 1}^{j_m^t -1}c_\ell^t \leq m-1 \leq  \sum_{\ell = 1}^{j_m^t}c_\ell^t.
\]
 
  \item Locate facility $m$ at the respective position of  node  $j_m^t$ on the 
  line $$x_m^t \leftarrow d(j,1) + \min_{p \in C_1 \cup \ldots \cup C_T \cup x^0} 
  p.$$
 \end{itemize}     
\end{itemize}
\caption{Algorithm for the offline case}\label{alg:offline}
\end{algorithm}
\footnotetext{the nodes $j_m^t$ can be equivalent calculated with the simpler 
criterion, \emph{$j_m^t$ is the most left
node with $f_{k j }^t > 0$}. See also Section~\ref{s:rounding_general}.}

Our algorithm, described in Algorithm~\ref{alg:offline}, is a simple
rounding scheme of the \emph{optimal fractional solution} of the Integer Linear 
Program of Figure~\ref{LP}.
This simple scheme produces an integral solution
that has the exact same cost with an  optimal fractional solution.

\begin{theorem}~\label{l:main_stageing}
Let $x$ denote the solution
produced by Algorithm~\ref{alg:offline}. Then
\[Cost(x) = \sum_{t=1}^T \bigg[ \sum\limits_{i\in C}\sum_{j\in V} 
d(\text{Loc}(i,t),j)x_{ij}^t +
\sum_{k\in F} S_k^t \bigg]\]
where $x_{ij}^t,S_k^t$ denote the values of these variables in the optimal 
fractional solution of the Integer Linear Program~(\ref{LP}).
\end{theorem}

Theorem~\ref{l:main_stageing} is the main result of this section and it
implies the optimality of our algorithm.
We remind that by Lemma~\ref{l:discrete}, there is an optimal
solution that locates facilities only in positions $C_1\cup \ldots \cup C_T \cup 
x^0$.
This solution corresponds to an integral solution of our Integer Linear Program,
meaning that $Cost(x^*)$ is greater than or equal to the cost of the 
\emph{optimal fractional solution},
which by Lemma~\ref{l:main_stageing} equals $Cost(x)$.
We dedicate the rest of the section to prove
Theorem~\ref{l:main_stageing}. The proof
is conducted in two steps and each step is exhibited in 
sections~\ref{s:semi_integral} and \ref{s:rounding_general} respectively.

In section~\ref{s:semi_integral}, we present a very simple rounding scheme in the 
case, where the values of the variables of the optimal fractional solution satisfy 
the following assumption.

\begin{assumption}\label{a:1}
Let $f_{jk}^t$ and $c_j^t$ be either $1/N$ or $0$, for some positive integer $N$.
\end{assumption}
\noindent Although Assumption~\ref{a:1} is very restrictive and its not generally 
satisfied, 
it is the key step for proving the optimality guarantee
of the rounding scheme presented in Algorithm~\ref{alg:offline}. Then, in 
section~\ref{s:rounding_general} we use the rounding scheme of 
section~\ref{s:semi_integral} to prove Theorem~\ref{l:main_stageing}.
In the upcoming sections,
$c_j^t,x_{ij}^t,f_{kj}^t,S_{kjl}^t,S_k^t$  will denote the values of
these variables in the \emph{optimal fractional} solution of the ILP~(\ref{LP}).

\subsection{Rounding Semi-Integral solutions}
\label{s:semi_integral}
Throughout this section, we suppose that Assumption~\ref{a:1} is satisfied; 
$f_{k j}^t$ and $c_j^t$ are either $1/N$ or $0$ , for some positive integer $N$. If 
the optimal fractional solution meets these requirements, then the integral 
solution presented in Lemma~\ref{l:first}
has the same overall cost. The goal of the section is to prove Lemma~\ref{l:first}.

\begin{definition}\label{d:positive}
$V_t^{+}$ denotes the set of nodes of $P$ with a positive amount of facility 
($c_j^t$) at stage $t$,
\[j \in V_{t}^+ \text{   if and only if    } c_j^t > 0\]
\end{definition}

\noindent We remind that since $c_j^t = 1/N$ or $0$, $|V_t^+| = K \cdot N$. We 
also consider the nodes in $V_t^{+} = \{Y_1^t,\ldots,Y_{K \cdot N}\}$ to be 
ordered from left to right. 

\begin{lemma}\label{l:first}
Let $Sol$ be the integral solution that at each stage $t$ places the $m$-th facility 
at the $(m-1)N + 1$ node of $V_t^+$
i.e. $Y_{(m-1)N + 1}^t$. Then, $Sol$ has the same cost
as the optimal fractional solution. 
\end{lemma}
\noindent The term \textbf{$\textbf{m}$-th facility} refers to the ordering of the 
facilities on the real line according to their initial positions 
$\{x_1^0,\ldots,x_K^0\}$. The proof of Lemma~\ref{l:first} is quite technically 
complicated, however it is based on two intuitive observations about the 
optimal fractional solution.

\begin{observation}\label{o:1}
The set of nodes at each \textbf{agent} $i$ \textbf{connects} at stage $t$ are 
\textbf{consecutive} nodes of $V_t^+$.
More precisely, there exists a set 
$\{Y_\ell^t,\ldots,Y_{\ell + N -1}^t\} \subseteq V_t^+$ such that

\[\sum_{j \in V}d(\text{Loc}(i,t),j)x_{ij}^t = \frac{1}{N}\sum_{h = \ell}^{\ell+ N 
-1}d(\text{Loc}(i,t),Y_h^t)\]

\end{observation}
\begin{proof}
Let an agent $i$ that at some stage $t$ has $x_{i Y_j^t}^t>0,x_{i Y_\ell^t}^t< 
1/N$ and $x_{i Y_h^t}^t > 0$ for some $j < \ell < h$. Assume that $\text{Loc}(i,t) 
\leq Y_\ell^t$ and to simplify notation consider $x_\ell = x_{i Y_\ell^t }^t,x_h = 
x_{i Y_h^t }^t $.
Now, increase $x_\ell$ by $\epsilon$ and decrease $x_h$ by $\epsilon$, where 
$\epsilon = \min(1/N - x_\ell,x_h)$. Then, the cost of the solution is decreased by 
$(d(\text{Loc}(i,t),h) - d(\text{Loc}(i,t),\ell))\epsilon > 0$, thus contradicting the 
optimality of the solution. The same argument holds if $\text{Loc}(i,t) \geq 
Y_\ell^t$. The proof follows since 
$\sum_{j \in V}x_{ij}^t = 1$.
\end{proof}

\begin{observation}\label{o:2}
Under Assumption~\ref{a:1},
the $m$-th facility places amount of facility $f_{m j}^t = 1/N$ from the $(m-1)N + 
1$ to the $m N$ node of $V_t^+$ i.e. to nodes $\{Y^t_{(m-1)N + 1},\ldots, Y_{m 
N}^t\}$.
\end{observation}

\noindent
Observation~\ref{o:2} serves in understanding the structure of the optimal 
fractional solution under Assumption~\ref{a:1}. However, it will be not used in 
this form in the rest of the section. We use Lemma~\ref{l:switching_cost} 
instead, which is
roughly a different wording of Observation~\ref{o:2} and its proof can be found 
in subsection~\ref{app:1} of the Appendix.

\begin{lemma}\label{l:switching_cost}
Let $S^t_k$ the 
fractional moving cost of facility $k$ at stage $t$. Then
\[\sum\limits_{t=1}^{T}\sum\limits_{k \in F} S^t_k =
\frac{1}{N}\sum\limits_{t=1}^{T}\sum\limits_{j=1}^{K\cdot N}d(Y_j^{t-1},Y_j^{t})\]
\end{lemma}

\noindent Observations~\ref{o:1}, and Lemma~\ref{l:switching_cost} 
(Observation~\ref{o:2}) are the key points in proving Lemma~\ref{l:first}.

\begin{definition}\label{d:Sol_p}
Let $Sol_p$ the integral solution that places at stage $t$ the $m$-th facility at 
the $(m-1)N+p$ node of $V_t^+$ i.e. $Y_{(m-1)N+p}^t$.
\end{definition}
\noindent Notice that the integral solution $Sol$ referred in Lemma~\ref{l:first} 
corresponds to $Sol_1$. The proof of Lemma~\ref{l:first} follows directly by 
Lemma~\ref{l:moving_cost} and Lemma~\ref{l:connection_cost}
that conclude this section.

\begin{lemma}\label{l:moving_cost}
Let $S_k^t$ be the moving cost of facility $k$ at stage $t$ in the optimal 
fractional solution and MovingCost$(Sol_p)$ the total moving cost of the 
facilities in the integral solution $Sol_p$. Then,

\[
\frac{1}{N}\sum_{p=1}^N MovingCost(Sol_p) = \sum_{t=1}^T\sum_{k\in F} S_k^t
\]
\end{lemma}

\begin{proof}
By the definition of the solutions $Sol_p$ we have that:

\begin{eqnarray*}
\frac{1}{N}\sum\limits_{p=1}^N MovingCost(Sol_p) &=& 
\frac{1}{N}
\sum\limits_{p=1}^N \sum\limits_{t=1}^T \sum_{m= 1}^K d(Y^{t-1}_{(m-1)N + 
p},Y^t_{(m-1)N + p})\\
&=& 
\frac{1}{N}
\sum\limits_{t=1}^T \sum_{m=1}^K\sum\limits_{p=1}^N 
d(Y^{t-1}_{(m-1)N+p},Y^t_{(m-1)N+p})\\
&=& 
\frac{1}{N}
\sum\limits_{t=1}^T \sum_{j=1}^{K \cdot N} d(Y^{t-1}_j,Y^t_j)\\
&=&
\sum_{t=1}^T\sum_{k\in F} S_k^t
\end{eqnarray*}
The last equality comes from Lemma~\ref{l:switching_cost}.
\end{proof}

\noindent Lemma~\ref{l:moving_cost}
states that if we pick uniformly at random one of the $N$ integral solutions 
$\{Sol_p\}_{p=1}^N$, then the expected moving cost that we will pay is equal to 
the moving cost paid by the optimal fractional solution. Interestingly, the same 
holds for the expected connection cost. This is formally stated in 
Lemma~\ref{l:connection_cost} and it is where Observation~\ref{o:1} comes into 
play.

\begin{lemma}\label{l:connection_cost}
Let $ConCost^t_i(Sol_p)$ denote the connection cost of agent $i$ at stage $t$ 
in $Sol_p$. Then,
\[\frac{1}{N}\sum_{p=1}^N ConCost^t_i(Sol_p) =\sum_{j\in V} 
d(\text{Loc}(i,t),j)x_{ij}^t
\]
\end{lemma}
As already mentioned, the proof of Lemma~\ref{l:connection_cost} crucially 
makes use of
Observation~\ref{o:1} and is presented in the subsection~\ref{app:1} of the 
Appendix.
Combining Lemma~\ref{l:moving_cost}
and~\ref{l:connection_cost} we get 
that if we pick an integral solution $Sol_p$ uniformly at random, the average 
total cost that we pay is $Z_{LP}^*$, where $Z_{LP}^*$ is the optimal fractional 
cost. More precisely,

\begin{eqnarray*}
\frac{1}{N} \sum_{p=1}^N Cost(Sol_p) &=&
\frac{1}{N} \sum_{p=1}^N [ MovingCost(Sol_p) + \sum_{t=1}^T\sum_{i \in C} 
ConCost_i^t(Solp)]\\
&=&  \sum_{t=1}^T[ \sum_{k=1}^K S_k^t + \sum_{i\in C}\sum_{j\in V} 
d(\text{Loc}(i,t),j)x_{ij}^t]\\
&=& Z_{LP}^*
\end{eqnarray*}
Since $Sol_p \geq Z_{LP}^*$, we have that 
$Sol_1 = \dots = Sol_N = Z_{LP}^*$ and this proves Lemma~\ref{l:first}.

\subsection{Rounding the General Case} \label{s:rounding_general}
In this section, we use Lemma~\ref{l:first} to prove 
Theorem~\ref{l:main_stageing}. As already discussed, Assumption~\ref{a:1} is 
not satisfied in general by the fractional solution of the linear program~(\ref{LP}).
Each $S_{k j \ell}^t$ will be either $0$ or $A_{k j \ell}^t/N_{k j \ell}^t$, for positive 
some integers $A_{k j \ell}^t, N_{k j \ell}^t$. However, each positive $f_{kj}^t$ 
will have the form $B_{k j}^t/N$,  where $N = \Pi_{S_{kj\ell}^t > 0} N_{k j \ell}^t$. 
This is due to the constraint $f_{k j}^t = \sum_{j \in V}S_{k j \ell}^t$.

Consider the path $P' = (V',E')$ constructed from
path $P= (V,E)$ as follows: Each node $j \in V$
is split into $K N$ \emph{copies} $\{j_1,\ldots, j_{KN}\}$ with zero distance 
between them. Consider also the LP~(\ref{LP}), when the underlying path is $P' = 
(V',E')$ and at each stage $t$, each agent $i$ is located
to a node of $V'$ that is a \emph{copy} of $i$'s original location, $Loc'(i,t) = \ell 
\in V'$, where $\ell \in \text{Copies}( \text{Loc}(i,t))$.
Although these are two different LP's, they are closely related since a solution 
for the one can be converted to a solution for the other with the exact same 
cost. This is due to the fact that
for all $j,h \in V$, $d(j,h) = d(j',h')$, where $j'\in \text{Copies}(j)$ and $h'\in 
\text{Copies}(h)$.

The reason that we defined $P'$ and the second LP is the following: Given an 
optimal fractional solution of the LP defined for $P$, we will construct a 
fractional solution for the LP defined for $P'$ with the exact same cost, which 
additionally satisfies Assumption~\ref{a:1}. Then, using Lemma~\ref{l:first} we 
can obtain an integral solution for $P'$ with the same cost. This integral solution 
for $P'$ can be easily converted to an integral solution for $P$. We finally show 
that these steps are done \emph{all at once }by the rounding scheme of 
Algorithm~\ref{alg:offline} and this concludes the proof of 
Theorem~\ref{l:main_stageing}.

Given the fractional positions $\{f_{k j}^t\}_{t \geq 1}$ of the optimal solution of 
the LP formulated for $P=(V,E)$, we construct the fractional positions
of the facilities in $P' = (V',E')$ as follows: 
If $f_{k j}^t = B_{k j}^t / N$, then 
facility $k$ puts a $1/N$ amount of facility in $B_{k j}^t$ nodes of the set 
$\text{Copies}(j)=\{j_1,\ldots, j_{K N}\}$
that have a $0$ amount of facility. The latter is possible since there are exactly 
$K N$ \emph{copies} of each $j \in V$ and $c_j^t \leq K$ (that is the reason we 
required $KN$ copies of each node). The values of the rest of the variables are 
defined in the proof of Lemma~\ref{l:construction} that is presented in the end of 
the section.
The key point is that the produced solution $\{f_{k \ell}^{'t},c_{j}^{'t}
,S_{k j \ell}^{'t}, x_{i j}^{'t},S_k^{'t}\}$
will satisfy the following properties (see Lemma~\ref{l:construction}):

\begin{itemize}
    \item its cost equals $Z_{LP}^*$

    \item $f^{'t}_{k\ell} = 1/N$ or $0$, for each $\ell \in V'$

    \item $c^{'t}_{\ell} = 1/N$ or $0,$ for each $\ell \in V'$
    
    \item $c^t_j = \sum\limits_{\ell \in \text{Copies}(j)}c^{'t}_{\ell}$,  for each $j \in 
    V$
\end{itemize}

\noindent Clearly, this solution satisfies Assumption~\ref{a:1} and thus 
Lemma~\ref{l:first} can be applied. This implies that the integral solution for $P'$
that places the $m$-th facility to the $(m-1)N + 1$ node of $V^{'+}_t$ 
($Y^{'t}_{(m-1)N + 1} \in V'$) has cost $Z_{LP}^*$. So the integral solution for 
$P$ that
places the $m$-th facility to the node $j_m^t \in V$, such that
$Y^{'t}_{(m-1)N + 1} \in \text{Copies}(j_m^t)$, has again cost $Z_{LP}^*$. 

A naive way to determine the nodes $j_m^t$ is to calculate $N$, construct $P'$ 
and its fractional solution, find the nodes $Y^{'t}_{(m-1)N + 1}$ and determine 
the nodes $j_m^t$ of $P$. Obviously, this rounding scheme requires exponential 
time. Fortunately, Lemma~\ref{l:fl} provides a linear time rounding scheme to 
determine the node $j_m^t$ given the optimal fractional solution of $P = (V,E)$. 
This concludes the proof of Theorem~\ref{l:main_stageing}.

\begin{lemma}\label{l:fl}
The $(m-1) N + 1$ node of $V_t'^+$ is a \emph{copy} of the node $j_m^t \in V$ 
if and only if
\[\sum_{\ell = 1}^{j_m^t-1} c_\ell^t \leq m-1 < \sum_{\ell = 1}^{j_m^t} c_\ell^t \]
.\end{lemma}
\begin{proof}

Let $(m-1) N + 1$ node of $V_t'^+$ be a \emph{copy} of the node $j_m^t \in 
V_t^+$.
Then 
\[\sum_{\ell = 1} ^{j_m^t -1}c_\ell ^t = \sum_{\ell = 1} ^{j_m^t -1} \sum_{\ell' \in 
\text{Copies}(\ell)} c_{\ell'}^{' \ t} 
\leq (m-1)N \frac{1}{N} = m-1\]

\[\sum_{\ell = 1} ^{j_m^t }c_\ell ^t =\sum_{\ell = 1} ^{j_m^t} \sum_{\ell' \in 
\text{Copies}(\ell)} c_{\ell'}^{' \ t} 
= ((m-1)N + 1) \frac{1}{N} > m-1
\]
\noindent The above equations hold because of the property $c_\ell^t = 
\sum_{\ell' \in \text{Copies}(\ell)}c_{\ell'}^{'t}$ and that $c_{\ell'}^{'t}$ is either $0$ 
or $1/N$.

\noindent Now, let $\sum_{\ell = 1}^{j_m^t-1} c_\ell^t \leq m-1 < \sum_{\ell = 
1}^{j_m^t} c_\ell^t$ and assume that $(m-1)N + 1$-th node of $V_t^+$ is a copy 
of $j \in V$.
If $j < j_m^t$, then $\sum_{\ell = 1}^{j}c_\ell^t> m-1$ and if $j > j_m^t$, then 
$\sum_{\ell = 1}^{j_m^t}< m-1$. As a result, $j = j_m^t$. 
\end{proof}

\begin{remark}
We remark that the nodes $j_m^t$ can be determined with an even simpler way 
than that presented in Algorithm~\ref{alg:offline}. That is 
$j_m^t$ is the most left node such that
$f_{m j}^t > 0$.
However, this rounding strategy requires some additional analysis.
\end{remark}

\begin{lemma}\label{l:construction}
Let $\{f_{kj}^t, c_j^t, S_{kjl}^t,x_{ij}^t\}_{t \geq 1}$ the optimal fractional solution 
for the LP~$(1)$ with underlying path $P$.
Then, there exists a solution $\{f_{k j}^{'t}, c_j^{'t}, S_{k j l}^{'t},x_{ij}^{'t},S_k^{'t} 
\}_{t \geq 1}$
of the LP~$(1)$ with underlying path $P'$ such that
\begin{enumerate}
    \item Its cost is $Z_{LP}^*$.

    \item $f^{'t}_{k\ell} = 1/N$ or $0$,  for each $\ell \in V'$

    \item $c^{'t}_{\ell} = 1/N$ or $0$,  for each $\ell \in V'$
    
    \item $c^t_j = \sum\limits_{\ell \in \text{Copies}(j)}c^{'t}_{\ell}$,  for each $j \in 
    V$
\end{enumerate}
\end{lemma}
\begin{proof}
 First, we set values to the variables $f_{kj}^{'t}$.
Initially, all $f_{kj}^{'t} = 0$.
We know that if $f_{kj}^t>0$, then it equals $B_{kj}^t/N$,
for some positive integer $B_{kj}^t$. For each
such $f_{kj}^t$, we find $u_1,\ldots,
u_{B_{kj}^t} \in \text{Copies}(j)$ with $f_{ku_h}^{'t} = 0$. Then, we set 
$f_{k u_h}^{'t} = 1/N$ for $h = \{1,B_{kj}^t\}$. Since there are $KN$ copies of 
each node $j \in V$ and $\sum_{j \in V}f_{kj}^t \leq K$, we can always find 
sufficient copies of
$j$ with $f_{ku}^{'t} = 0$. When this step is terminated, we are sure that 
conditions $2,3,4$ are satisfied. 

We continue with the variables $S_{kj\ell}^{'t}$.
Initially, all $S_{kj\ell}^{'t} = 0$. Then, each positive
$S_{kj\ell}^{t}$ has the form $B_{kj\ell}^{t}/N$.
Let $B = B_{kj\ell}^{t}$ to simplify notation.
We now find $B$ copies of $u_1,\ldots, u_{B}$ of $j$ and $v_1,\ldots, v_B$ of 
$\ell$ so that

\begin{itemize}
    \item $f_{ku_1}^{' t} = \dots = f_{ku_{B}}^{'t} = f_{kv_1}^{'t} = \dots = 
    f_{kv_{B}}^{'t} = 1/N$

 \item $S_{ku_1 h }^{'t} = \dots = S_{k u_{B} h}^{'t} = S_{k h v_1 }^{'t} = \dots = 
 S_{k h v_{B}}^{'t} = 0$
for all $h \in V'$

\end{itemize}
We then set 
$S_{ku_1 v_1 }^{'t} = \dots = S_{ku_B v_B }^{'t} = 1/N$. Again, since $\sum 
_{\ell \in V}S_{kj \ell }^{t} = f_{kj}^t$ and $\sum _{j \in V}S_{kj \ell }^{t} = 
f_{k\ell}^t$ we can always find $B_{kj\ell}^t$ pairs of copies of $j$ and $\ell$ that 
satisfy the above requirements. We can now prove that the movement cost
of each facility $k$ is the same in both solutions.
\begin{eqnarray*}
\sum_{j \in V}\sum_{\ell \in V}d(j,\ell)S_{kj\ell}^t &=&
\sum_{j \in V}\sum_{\ell \in V}d(j,\ell)B_{kj\ell}^t/N\\
&=&
\sum_{j \in V}\sum_{\ell \in V}\sum_{h \in \text{Copies}(j)}\sum_{ h' \in 
\text{Copies}(\ell)}S_{k h h'}^{'t}d(h,h')\\
&=&
\sum_{j' \in V'}\sum_{\ell' \in V'} S_{k j' \ell'}^{'t}d(j',\ell')
\end{eqnarray*}
\noindent The second equality follows from the fact that $h,h'$ are copies of 
$j,\ell$ respectively and thus 
$d(h,h')=d(j,\ell)$.

Finally, set values to the variables $x_{ij}^{'t}$ for each $j \in V'$. Again, each 
positive $x_{ij}^t$ equals
$B_{ij}^t/N$, for some positive integer. We take $B_{ij}^t$ copies of $j$, 
$u_1,\dots,u_{B_{ij}^t}$ and
set $x_{iu_1}^{'t} = \dots = x_{iu_{B_{ij}^t}}^{'t} = 1/N$. The connection cost of 
each agent $i$ remains the same since 

\begin{eqnarray*}
\sum_{j \in V} d(\text{Loc}(i,t),j)x_{ij}^t &=&
\sum_{j \in V} d(\text{Loc}(i,t),j)B_{ij}^t/N\\
&=&
\sum_{j \in V} d(\text{Loc}(i,t),j)\sum_{j' \in \text{Copies}(j)}x_{i j'}^{'t}\\
&=&
\sum_{j \in V} \sum_{j' \in \text{Copies}(j)}
d(\text{Loc}'(i,t),j') x_{i j'}^{'t}\\
&=&
\sum_{h \in V'} d(\text{Loc}'(i,t),h) x_{i h}^{'t}\\
\end{eqnarray*}
The third equality holds since $\text{Loc}'(i,t) \in \text{Copies}(\text{Loc}(i,t))$.
\end{proof}

\section{A Constant-Competitive Algorithm for the Online 2-Facility Reallocation 
Problem}\label{s:online}
\tikzset{cross/.style={path picture={
  \draw[black]
    (path picture bounding box.south east)--(path picture bounding
box.north west)
    (path picture bounding box.south west)--(path picture bounding
box.north east);
}}}

In this section, we present an algorithm for the online \emph{2-facility 
reallocation problem} and
we discuss the core ideas that prove its performance guarantee. The online 
algorithm, denoted as
Algorithm~\ref{alg:double_coverage}, consists of two major steps.

In Step~$1$, facilities are initially moved towards the positions of the agents.
We remark that in Step~$1$, the final positions of the facilities at stage $t$ are 
not yet determined.
The purpose of this step is to bring at least one facility close to the agents. 
This initial moving consists of three cases 
(see Figure~\ref{f:1}), depending only on the relative positions of the facilities at 
stage $t-1$ and the agents at stage $t$.

In Step~$2$, our algorithm determines the final positions of the facilities 
$x_1^t,x_2^t$.
Notice that after Step~$1$, at least one of the facilities is inside the interval
$[\alpha_1^t,\alpha_n^t]$, meaning that at least one of the facilities is close to 
the agents.
As a result, our algorithm may need to decide between moving the second 
facility close to the agents or
just letting the agents connect to the facility that is already close to them. 
Obviously, the first choice
may lead to small connection cost, but large moving cost, while the second has 
the exact opposite effect.
Roughly speaking, Algorithm~\ref{alg:double_coverage} does the following:
If the connection cost of the agents, when placing just one facility optimally, is 
not much greater
than the cost for moving the second facility inside $[\alpha_1^t,\alpha_n^t]$,
then Algorithm~\ref{alg:double_coverage} puts the first facility to the position 
that minimizes
the connection cost, if one facility is used. Otherwise,
it puts the facilities to the positions that minimize the connection cost, if two 
facilities are used.
The above cases are depicted in Figure~\ref{f:2}.
We formalize how this choice is performed, introducing some additional notation.

\begin{definition}\label{d:online}
	\mbox{}
 \begin{itemize}
  \item $C_t = \{\alpha_1^t,\ldots, \alpha_n^t\}$ denotes the positions of the 
  agents at stage $t$ \text{ordered from left to right}.

  \item If $C$ is a set of positions with $|C|=2k, k \in \mathbb{N}_{>0}$, then 
  $M_C$ denotes the median interval
  of the set, which is the interval $[\alpha_{n/2}, \alpha_{n/2+1}]$. If $|C| = 2k +1, k 
  \in \mathbb{N}_0$, then $M_C$ is a single point.
 
  \item $H(C)$ denotes the optimal connection cost for the set $C$ when all 
  agents of $C$
  connect to just one facility. That is
  $H(C) = \sum_{\alpha \in C}|\alpha - M_C|.$ We also define
  $H(\emptyset)=0$.


  \item $C_{1t}^*$ (resp. $C_{2t}^*$) denotes the positions of the agents that 
  connect to facility $1$ (resp. $2$) at stage $t$
  in the optimal solution $x^*$. $C_{1t}$ (resp. $C_{2t}$) denotes the positions of 
  the agents that connect to facility $1$ (resp. $2$) at stage $t$
  in the solution produced by Algorithm~\ref{alg:double_coverage}.
 \end{itemize}
\end{definition}

\begin{algorithm}[ht]
	\SetAlgoLined
	\DontPrintSemicolon
	\KwData{At stage $t \geq 1$ the new positions of the agents $C_t = 
		\{\alpha_1^t,\ldots,\alpha_n^t\}$,
		ordered from left to right, arrive}
	\textbf{Step 1:} Moving the facilities towards the agents
	
	\begin{algorithmic}
		\STATE $z_1 \leftarrow x_1^{t-1}$, $z_2 \leftarrow x_2^{t-1}$
		
		\STATE \If{$z_1>\alpha_n^t$}{\emph{move facility $1$ to the left until it hits 
			}$\alpha_n^t$\;
			$z_1 \leftarrow \alpha_n^t$}
		
		\STATE \If{$z_2<\alpha_1^t$}{\emph{ move facility $2$ to the right until it 
				hits }$\alpha_1^t$\;
			$z_2 \leftarrow \alpha_1^t$}
		\STATE \If{$z_1<\alpha_1^t$ \textbf{and} $z_2>\alpha_n^t$}
		{
			\emph{move facility $1$ to the right and facility $2$ to the left
				until a facility hits }$[\alpha_1^t,\alpha_n^t]$\;
			$z_1\leftarrow z_1 + \min(|x_1^{t-1}-\alpha_1^t|,|x_2^{t-1}-\alpha_n^t|)$\;
			$z_2\leftarrow z_2 - \min(|x_1^{t-1}-\alpha_1^t|,|x_2^{t-1}-\alpha_n^t|)$\;
		}
	\end{algorithmic}
	\textbf{Step 2:} Selecting the final position of the facilities
	
	\begin{algorithmic}
		\STATE \If{$\alpha_1^t\leq z_1 \leq \alpha_n^t$ \textbf{and} $z_2- 
			\alpha_n^t\geq 3H(C_t)$}{
			\emph{put facility $1$ to the median of }$C_t$
			\emph{ and move facility }$2$\emph{ to the left by } $3H(C_t)$\;
			$x_1^t\leftarrow M_{C_t}$,$~x_2^t\leftarrow z_2 - 3H(C_t)$\;
		}
		
		\STATE \If{$\alpha_1^t\leq z_2 \leq \alpha_n^t$ \textbf{and} $\alpha_1^t- 
			z_1\geq 3H(C_t)$}{
			\emph{put facility $2$ to the median of }$C_t$
			\emph{ and move facility }$1$\emph{ to the right by } $3H(C_t)$\;
			$x_1^t\leftarrow z_1 + 3H(C_t),x_2^t\leftarrow M_{C_t}$\;
		}
		\STATE     \Else{
			
			\emph{Compute the optimal partition} $(O_1,O_2)$\emph{ of }$C_t$
			\emph{ that minimizes the connection cost at stage} $t$.\;
			
			\emph{Put facility }$1$\emph{ to the median of }$O_1$\emph{ and
				facility }$2$\emph{ to the median of } $O_2$.\;
			$x_1^t\leftarrow M_{O_1} , x_2^t\leftarrow M_{O_2}$
		}
	\end{algorithmic}
	\caption{Selecting $x_1^t$ and $x_2^t$}
	\label{alg:double_coverage}
\end{algorithm}
\setlength{\belowcaptionskip}{0pt}

We first mention that Algorithm~\ref{alg:double_coverage} seems much more 
complicated than
it really is (the first two cases are symmetric both in Step~1 and Step~2). In fact,
only the last two cases are difficult to handle and we explain them subsequently.
The performance guarantee of  Algorithm~\ref{alg:double_coverage} is formally 
stated in Theorem~\ref{t:competitive}.

\begin{theorem}\label{t:competitive}
Let $x = \{x_1^t,x_2^t\}_{t\geq 1}$ the solution
produced by Algorithm~\ref{alg:double_coverage} and $x^*$ the optimal 
solution. Then,
\[Cost(x) \leq 63 \cdot Cost(x^*) + |x_1^0 - x_2^0|\]
where $x_1^0,x_2^0$ are the initial positions of the facilities.
\end{theorem}

The rest of the section is dedicated to provide a proof sketch (some proofs are 
included in subsection~\ref{app:2} of the Appendix) of 
Theorem~\ref{t:competitive}. Although it is possible to improve the competitive 
ratio of Algorithm~\ref{alg:double_coverage} by a much more technically 
involved analysis, we stress here that it is not possible to turn the result into any 
constant factor. The reason is that the \emph{2-facility reallocation problem} on 
the line is a generalization of \emph{2-server problem} on the line, which has a 
lower bound of 2 on the competitive ratio of any online algorithm. 
Before proceeding, we present Lemma~\ref{l:median} that is a key component in 
the subsequent
analysis and that reveals the real difficulty of the online \emph{$2$-facility 
reallocation problem}.

\begin{lemma}\label{l:median}
Let the optimal solution $x^*$ and $C_{1t}^*,C_{2t}^*$
the set of agents that connect respectively to facility $1$ and $2$ at stage $t$.
Let the solution $y^t=(y_1^t,y_2^t)$ defined as follows:
$$
y_k^t= \left\{
\begin{array}{ll}
      M_{C_{kt}^*} & \text{ if  } C_{kt}^* \neq \emptyset \\
      x_k^{*t} & \text{ if  } C_{kt}^* = \emptyset\\
\end{array}
\right.
$$
Then, the following inequality holds.
\[ \sum_{t=1}^T \bigg[ \sum_{k=1}^2  [ H(C_{kt}^*) + |y_k^t - y_k^{t-1}| ] \bigg] 
\leq 3\cdot Cost(x^*)\]
\end{lemma}

\begin{proof}
Since $\sum_{t=1}^T \sum_{k=1}^2  H(C_{kt}^*)=
\sum_{t=1}^T \sum_{k=1}^2 \sum_{a \in C_{kt}^*} |x_k^{*t} - a|,$
we only have to prove that $$ 
\sum_{t=1}^T \sum_{k=1}^2  |y_k^t - y_k^{t-1}| \leq 2\sum_{t=1}^T \sum_{k=1}^2 
\left[H(C_{kt}^*)  +|x_k^{*t} - x_k^{*t-1}|\right]$$
From the triangle inequality, we have that
$$\sum_{t=1}^T \sum_{k=1}^2  |y_k^t - y_k^{t-1}|\leq \sum_{t=1}^T 
\sum_{k=1}^2 [
|y_k^t - x_k^{*t}|+
|y_k^{t-1} - x_k^{*t-1}|+|x_k^{*t} - x_k^{*t-1}|]$$
The right hand side of the inequality is maximized, when $y_k^t\neq x_k^{*t}$ 
and $y_k^{t-1}\neq x_k^{*t-1}$ for $k=1,2$, namely when $C_{kt}^*, 
C_{k(t-1)}^*\neq \emptyset$. Since $y_k^t$ (resp. $y_k^{t-1}$) is the median 
agent (lies in the median interval of  $C_{kt}^*$ in the case $|C_{kt}^*|=2k$) of 
$C_{kt}^*$ (resp.$C_{k(t-1)}^*$)  in this case, 
$x_k^{*t}$ (resp. $x_k^{*t-1}$) has to connect $y_k^t$ (resp. $y_k^{t-1}$).
Thus, $\sum_{t=1}^T \sum_{k=1}^2  |y_k^t - x_k^{*t}|+
|y_k^{t-1} - x_k^{*t-1}|$ can be upper bounded by the optimal connection cost, 
which is 

$$ \sum_{t=1}^T \sum_{k=1}^2 \sum_{a \in C_{kt}^*} |x_k^{*t} - a|+\sum_{t=1}^T 
\sum_{k=1}^2 \sum_{a \in C_{k(t-1)}^*} |x_k^{*t-1} - a| 
\leq 2\sum_{t=1}^T \sum_{k=1}^2 H(C_{kt}^*)$$

\end{proof}


Lemma~\ref{l:median} indicates that the real difficulty of the problem is not 
determining the exact positions
of the facilities in the optimal solution, but to determine the \emph{service 
clusters} that the optimal solution forms. In fact,
if we knew the clusters $C_{1t}^*,C_{2t}^*$, then Lemma~\ref{l:median} provides 
us with a $3$-approximation algorithm.
Obviously, this information cannot be acquired in the online setting, since 
$C_{1t}^*,C_{2t}^*$ depend on the future positions
of the agents that
we do not know. We prove that Algorithm~\ref{alg:double_coverage} has an 
approximation guarantee of $21$ with respect to
 the solution $y$, that directly translates to an approximation guarantee of $63$ 
 with respect to
$Cost(x^*)$. The latter is formally stated in Lemma~\ref{l:competive_with_y} and 
is the main result of this section.

\begin{lemma}\label{l:competive_with_y}
Let $x=\{x_1^t,x_2^t\}_{t \geq 1}$ be the solution produced by 
Algorithm~\ref{alg:double_coverage}. Then,
the cost paid by solution $x$ at stage $t$,
$\sum_{k=1}^2|x_k^t - x_k^{t-1}| + \sum_{i=1}^n \min_{k=1,2}|x_k^t - 
\alpha_i^t|$, is at most
\[21\sum_{k=1}^2 [ H(C_{kt}^*) + |y_k^t - y_k^{t-1}|]   + \Phi_t(x^t) - 
\Phi_{t-1}(x^{t-1})\]
\text{where} $\Phi_t(x_1,x_2) = 2(|x_1 - y_1^t| +  |x_2 - y_2^t|) + |x_1 - x_2|$.
\end{lemma}

Lemma~\ref{l:competive_with_y} directly implies Theorem~\ref{t:competitive} by 
applying
a telescopic sum over all $t$ and then applying Lemma~\ref{l:median}.  Notice 
that the additive term $|x_1^0 - x_2^0|$ in Theorem~\ref{t:competitive} depends 
only on the initial positions of the facilities and follows from
the fact that $\Phi_0(x^0) = |x_1^0 - x_2^0|$. Since the additive term is a 
constant independent from the request sequence (the client positions $C_t$), it 
is common to define the competitive ratio of an online algorithm as in 
Section~\ref{sec:prelim}.
In the rest of the section, we present the proof ideas of 
Lemma~\ref{l:competive_with_y}, which come together with explaining 
Steps~$1$ and~$2$ of our algorithm. 
Let us start with explaining Step~$1$. First, note that since $x_1^0\leq x_2^0$, 
then $x_1^t\leq x_2^t$ by our algorithm construction. Now, assume that 
$x_2^{t-1} \leq \alpha_1^t$ (second case).
Before deciding the exact positions of the facilities, we can \emph{safely} move 
facility $2$ to the right until reaching
$\alpha_1^t$. The term \emph{safely} means that this moving cost is 
\emph{roughly} upper bounded by the moving cost $\sum_{k=1}^2|y_k^t - 
y_k^{t-1}|$.
This \emph{safe moving} applies to all three cases of Step~$1$ in 
Algorithm~\ref{alg:double_coverage}
and is formally stated in Lemma~\ref{l:z1}.

\begin{figure}[t]
\centering
\begin{tikzpicture}[scale = 1]
\draw (-5.5,0) -- (4.5,0);
\coordinate (y1p) at (-4.5,0);
\coordinate (y1a) at (-2,0);
\coordinate (y2p) at (-3,0);
\coordinate (y2a) at  (-1,0);
\coordinate (y1O) at  (-3,0);
\coordinate (y2O) at  (1,0);
\coordinate (a1) at (-1,0);
\coordinate (a2) at (2,0);

 \filldraw (y1p) node[align=center,   below] {$x_1^{t-1} $} --
(y1p)  node[cross,thick,minimum size=4pt] {}     --
(y2p) node[align=center,  below] {$x_2^{t-1}$} --
(y2p)  node[cross,thick,minimum size=4pt] {}     --
(a1) circle (2pt) node[align=center, above] {$\alpha_1^t$} --
(-0.7,0) circle (2pt)  --
(0.2,0) circle (2pt)   --
(0.8,0) circle (2pt)  --
(1.3,0) circle (2pt)  --
(a2) circle (2pt) node[align=center, above] {$\alpha_n^t$} ;
\draw[->,red,thick] (y2p) to  (y2a);
\node[text width = 8cm] at (0,-1.7){\small If both facilities $1$ and $2$ are
on the left of the agents, then facility $2$ is moved to the right
until hitting the position of the leftmost agent (the case with facilities $1$ and $2$
on the right of agents is symmetric).
};
\end{tikzpicture}

\begin{tikzpicture}[scale = 1]
\draw (-5.5,0) -- (4.5,0);
\coordinate (y1p) at (-4.5,0);
\coordinate (y1a) at (-2.5,0);
\coordinate (y2p) at (4,0);
\coordinate (y2a) at  (2,0);
\coordinate (y1O) at  (-3,0);
\coordinate (y2O) at  (1,0);
\coordinate (a1) at (-1,0);
\coordinate (a2) at (2,0);

 \filldraw (y1p) node[align=center,   below] {$x_1^{t-1} $} --
 (-2.5,0) node[align=center,   below=3pt] {} --
(y1p)  node[cross,thick,minimum size=4pt] {}     --
(y1a)  node[cross,thick,minimum size=4pt] {}     --
(y2p) node[align=center,  below] {$x_2^{t-1}$} --
(y2p)  node[cross,thick,minimum size=4pt] {}     --
(a1) circle (2pt) node[align=center, above] {$\alpha_1^t$} --
(-0.3,0) circle (2pt)  --
(0.2,0) circle (2pt)   --
(0.8,0) circle (2pt)  --
(1.3,0) circle (2pt)  --
(a2) circle (2pt) node[align=center, above] {$\alpha_n^t$} --
(y2a)  node[cross,thick,minimum size=4pt] {}     --
(y2a) node[align=center, below=5pt] {} ;

\draw[->,red,thick] (y1p) to (y1a);
\draw[->,red,thick] (y2p) to (y2a);
\node[text width = 8cm] at (0,-1.7){\small If facility $1$ is on the left of the agents
and facility $2$ is on the right of the agents, then both facilities are moved with 
the same speed
towards the interval $[\alpha_1^t,\alpha_n^t]$ until one of them hits the interval.
};

\end{tikzpicture}

\caption{Step~$1$ of Algorithm~\ref{alg:double_coverage} is depicted. After this 
step, the positions of the facilities are denoted by $z_1,z_2$ in 
Algorithm~\ref{alg:double_coverage}.}
\label{f:1}
\setlength{\belowcaptionskip}{0pt}
\end{figure}

\begin{lemma}\label{l:z1}
Let $z = (z_1,z_2)$ denote the values of the variables $z_1,z_2$ after Step~1 of 
Algorithm~\ref{alg:double_coverage}. Then,
\[\sum_{k=1}^2|z_k - x_k^{t-1}| \leq 2\sum_{k=1}^2 |y_k^t - y_k^{t-1}| -\Phi_t(z) 
+\Phi_{t-1}(x^{t-1})\]
\end{lemma}

The proof of Lemma~\ref{l:z1} can be found in subsection~\ref{app:2} of the 
Appendix. Lemma~\ref{l:z1} reveals the basic idea of Step~$1$ performed by the 
online algorithm. We remind that Step~$1$ is performed by 
Algorithm~\ref{alg:double_coverage} if both facilities are outside the interval 
$C_t$ at the beginning of stage $t$. Therefore, it distinguishes between the three 
cases depicted in Figure~\ref{f:1} (We show 2 cases since the case with both 
facilities on the right of the agents is symmetric to the first). Moving with the 
same speed towards the interval $[a_1^t, a_n^t]$ results to the same moving 
cost for both facilities; both facilities will move the distance of the facility which 
is closest to its closest agent.

According to the \emph{geometry} of the agents' positions, we can identify a 
\emph{safe move} whose cost is also paid by solution $y$ for moving the 
facilities.
Moreover, the proof of Lemma~\ref{l:z1} reveals why we compare our algorithm
with the solution $y$ and not directly with $x^*$. All these \emph{safe moves} 
are based on the fact that
either $y_1^t$ or $y_2^t$ lies in the $C_t=[\alpha_1^t,\alpha_n^t]$ (the latter 
does not necessarily hold for $x^*$).
Finally, the \emph{potential function} $\Phi_t(x_1,x_2)$
is crucial, since it permits \emph{safe moves}, when all agents are on the 
right/left of the facilities (first/second case)
as well as when they are contained in the interval $[x_1^{t-1},x_2^{t-1}]$ (third 
case). This idea
was first developed for the $K$-server problem \citet{K2009}.

Up next, we analyze the ideas of Step~$2$. We now need to bound the 
connection cost plus some  additional moving cost from the point where 
\emph{the safe move stopped}. The following lemma formalizes the guarantees 
provided by Algorithm~\ref{alg:double_coverage} after the execution of 
Step~$2$. 
The full proof of Lemma~\ref{l:z2} can be found in subsection~\ref{app:2} of the 
Appendix.

\begin{lemma}\label{l:z2}
\emph{Let $x^t = (x_1^t,x_2^t)$ denote the locations of facilities at stage $t$ 
after the execution of Step~2.
Then,} \[ \sum_{k=1}^2 [H(C_{kt}) + |x^t_k - z_k|]  \leq 21\sum_{k=1}^2H(C_{kt}^*) 
-\Phi_t(x^t) + \Phi_t(z)\]
\end{lemma}

When Algorithm~\ref{alg:double_coverage} performs Step~$2$, we know that at 
least one facility lies inside the interval $C_t$.  This facility will definitely connect 
some agents of $C_t$, since we can charge it a small moving cost even if it 
connects all agents. Thus, the algorithm needs only to decide whether it will 
connect agents by using only one facility or by using both facilities.  The 
decision depends on the distance between the facility, which is outside of $C_t$ 
(in case there is one), and the closest agent to this facility. If this distance is 
"small" (resp. if the facility is already inside the interval), 
Algorithm~\ref{alg:double_coverage}
will connect agents to both facilities minimizing the connection cost using two 
facilities. This will guarantee that the moving cost and connection cost incurred 
are relatively small compared to the cost of solution $y$.

Now, if the facility, which is outside the interval $C_t$, is "far" from its closest 
agent, Algorithm~\ref{alg:double_coverage}
moves this facility towards $C_t$ by a distance, depending on the optimal 
connection cost using one facility, and serves all agents using the facility, which 
is inside $C_t$. Then, we can prove that this move is sufficient to bound the 
total cost of the algorithm compared to the cost of solution $y$, even if $y$ has 
arbitrarily smaller connection cost (if it uses both facilities to serve the agents). 
The choices of  Algorithm~\ref{alg:double_coverage} are depicted in 
Figure~\ref{f:2}. We provide more detailed Figures based on the  analysis of  
Algorithm~\ref{alg:double_coverage} in subsection~\ref{app:2} of the Appendix.

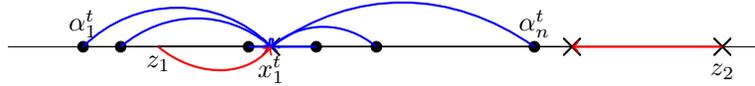
\begin{figure}[ht]
\centering
\begin{tikzpicture}[scale = 1]
\draw (-5.5,0) -- (4.5,0);
\coordinate (y1p) at (-3.5,0);
\coordinate (y1a) at (-2,0);
\coordinate (y2p) at (4,0);
\coordinate (y2a) at  (2,0);
\coordinate (y1O) at  (-2,0);
\coordinate (y2O) at  (2,0);
\coordinate (a1) at (-4.5,0);
\coordinate (a2) at (1.5,0);
\filldraw (y1p) node[align=center,   below=4pt] {} --
(y1a)  node[align=center, below] {$x_1^t$}     --
(y1p)  node[align=center, below] {$z_1$}     --
(y1a)  node[cross,thick,minimum size=4pt] {}     --
(y2p) node[align=center,  below=3pt] {$z_2$} --
(y2p)  node[cross,thick,minimum size=3pt] {}     --
(a1) circle (2pt) node[align=center, above] {$\alpha_1^t$} --
(a2) circle (2pt) node[align=center, above] {$\alpha_n^t$} --
(-4,0) circle (2pt)  --
(-2.3,0) circle (2pt)   --
(-1.4,0) circle (2pt)  --
(-0.6,0) circle (2pt)  --
(y2a)  node[cross,thick,minimum size=4pt] {}  ;
\draw[->,blue,thick] (a1) to[out=45,in=135] (y1a);
\draw[->,blue,thick] (-4,0) to[out=40,in=140] (y1a);
\draw[->,blue,thick] (-2.3,0) to (y1a);
\draw[->,blue,thick] (-1.4,0) to (y1a);
\draw[->,blue,thick] (-0.6,0) to[out=140,in=40] (y1a);
\draw[->,blue,thick] (a2) to[out=145,in=35] (y1a);
\draw[->,red,thick] (y1p) to[out=320,in=230]  (y1a);

\draw[->,red,thick] (y2p) to  (y2a);
\node[text width = 8cm] at (0,-1.7){ \small The first choice of Step~2 is depicted. 
In this case,
the facility initially lying inside the interval $[\alpha_1^t,\alpha_n^t]$ moves to the 
median of agents. In this position, the connection cost is minimized using one 
facility.
};
\end{tikzpicture}

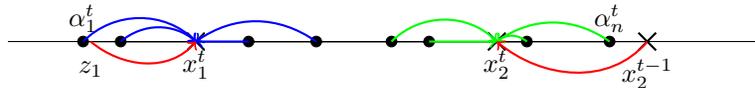
\begin{tikzpicture}[scale = 1]
\draw (-5.5,0) -- (4.5,0);
\coordinate (y1p) at (-4.4,0);
\coordinate (y1a) at (-3,0);
\coordinate (y2p) at (3,0);
\coordinate (y2a) at  (1,0);
\coordinate (a1) at (-4.5,0);
\coordinate (a2) at (2.5,0);
\filldraw (y1p) node[align=center,   below=4pt] {$z_1$} --
(y1a)  node[align=center, below] {$x_1^t$}     --
(y1a)  node[cross,thick,minimum size=4pt] {}     --
(y2p) node[align=center,  below=3pt] {$x_2^{t-1}$} --
(y2p)  node[cross,thick,minimum size=3pt] {}     --
(a1) circle (2pt) node[align=center, above] {$\alpha_1^t$} --
(a2) circle (2pt) node[align=center, above] {$\alpha_n^t$} --
(-4,0) circle (2pt)  --
(-2.3,0) circle (2pt)   --
(-1.4,0) circle (2pt)  --
(-0.4,0) circle (2pt)   --
(0.1,0) circle (2pt)   --
(1.4,0) circle (2pt)  --
(y2a) node[align=center, below] {$x_2^t$}    --
(y2a)  node[cross,thick,minimum size=4pt] {}  ;
\draw[->,blue,thick] (a1) to[out=45,in=135] (y1a);
\draw[->,blue,thick] (-4,0) to[out=40,in=140] (y1a);
\draw[->,blue,thick] (-2.3,0) to (y1a);
\draw[->,blue,thick] (-1.4,0) to[out=145,in=35] (y1a);
\draw[->,green,thick] (-0.4,0) to[out=45,in=135] (y2a);
\draw[->,green,thick] (0.1,0) to (y2a);
\draw[->,green,thick] (1.4,0) to[out=140,in=40] (y2a);
\draw[->,green,thick] (a2) to[out=145,in=35] (y2a);
\draw[->,red,thick] (y1p) to[out=320,in=230]   (y1a);

\draw[->,red,thick] (y2p) to[out=230,in=320]  (y2a);

\node[text width = 8cm] at (0,-1.7){\small The second choice of Step~2 is 
depicted. 
Facilities are placed to the positions, where the connection cost of the agents is 
minimized using two facilities.
};
\end{tikzpicture}
\caption{Step~$2$ of Algorithm~\ref{alg:double_coverage} is depicted.}
\label{f:2}
\end{figure}

\section{Open Problems}
Regarding the offline variant of the 
$K$-\emph{facility reallocation} problem, it would be interesting to consider the 
problem in general metric spaces. Since $K$-\emph{facility reallocation}  is 
essentially a dynamic $K$-\emph{median} problem, a main open problem is to 
design approximation algorithms for this problem as well as to find lower 
bounds on the approximation ratio of any offline algorithm in general metric 
spaces. Turning to the online variant, the main question arising is to design an 
online algorithm for online $K$-\emph{facility reallocation} problem on the line. 
This variant with any number of facilities seems much more intriguing. It would 
also be interesting to consider randomized algorithms for both the online and the 
offline variant.   
\bibliographystyle{plainnat}
\bibliography{1}

\appendix
\section{Appendix}
\begin{subsection}{Omitted proofs of 
Section~\ref{s:rounding_general}}\label{app:1}

\textbf{Lemma \ref{l:switching_cost}}
\emph{Let $S^t_k$ the 
fractional switching cost of facility $k$ at stage $t$. Then,}
\[\sum\limits_{t=1}^{T}\sum\limits_{k \in F} S^t_k =
\frac{1}{N}\sum\limits_{t=1}^{T}\sum\limits_{j=1}^{K\cdot N}d(Y_j^{t-1},Y_j^{t})\]

\begin{proof}
By Assumption~\ref{a:1}, $c_j^t=1/N$ if $j\in V_t^+= \{Y_1^t,\ldots,Y_{KN}^t\}$
and $0$ otherwise. Notice that the connection cost of the optimal fractional 
solution only depends on the
variables $c_j^t$. As a result, $f_{k j}^t,S_k^t,S_{k j l}^t$ must be the optimal 
solution of the following linear program.
$$
\begin{array}{ll@{}ll}
\text{minimize}  &  \sum \limits_{t=1}^{T}\sum \limits_{k=1}^KS_k^t\\
& \\
\text{s.t.}& \sum\limits_{k \in F }f^t_{kj}=\frac{1}{N}   & \forall j \in V_t^+,t\in \{1,T\} 
\\
        &\sum \limits_{j \in V_t^+} f_{kj}^t=1 & \forall k\in F, t\in \{1,T\}\\
        &S_k^t = \sum \limits_{j,l \in V} d(j,l)S_{kjl}^t & \forall k \in F, t\in \{1,T\} \\
        &\sum \limits_{j \in V_{t-1}^+} S_{kjl}^t =f_{kl}^t & \forall k \in F,l \in 
        V_t^+,t\in \{1,T\}\\
        &\sum \limits_{l \in V_t^+} S_{kjl}^t = f_{kj}^{t-1} & \forall k \in F,j \in 
        V_{t-1}^+,t\in \{1,T\}\\
\end{array}
$$

\noindent Instead of proving that the minimum cost of the above linear program
is $\frac{1}{N}\sum\limits_{t=1}^{T}\sum\limits_{j=1}^{K\cdot N}d(Y_j^{t-1},Y_j^t)$,
we prove this for the following more convenient relaxation of the above LP. 
\begin{equation}\label{eq:LP2}
\begin{array}{ll@{}ll}
\text{minimize}  &  \sum \limits_{t=1}^T\sum \limits_{j \in V_{t-1}^+,l \in V_t^+} 
d(j,l)F_{jl}^t\\
& \\
\text{s.t.}& \sum\limits_{l \in V_t^+}F^t_{jl}=\frac{1}{N}   & \forall j \in V_{t-1}^+,t\in 
\{1,T\} \\
        &\sum \limits_{j \in V_{t-1}^+} F^t_{jl}=\frac{1}{N} & \forall l\in V_t^+, t\in 
        \{1,T\}\\
\end{array}
\end{equation}
It is easy to prove that the LP~(\ref{eq:LP2}) is a relaxation of the first 
by setting $F_{jl}^t = \sum _{k \in F}S_{kjl}^t$. Moreover, the above LP 
describes a flow 
problem between the nodes $V_t^+$, where $F_{jl}^t$ is the amount of flow 
going from node $j \in V_{t-1}^+$
to node $l \in V_t^+$ (see Figure~\ref{f:flow}).

We are ready for the final step of our proof. First, observe that 
$F^t_{Y_j^{t-1}Y_j^t}$ is feasible
solution for the above LP since $|V_{t-1}^+|=|V_t^+|=K\cdot N$. If we prove that 
this assignment minimizes the objective, then we are done. Assume that in the 
optimal solution, $F^t_{Y_1^{t-1}Y_1^t}<1/N$. 
Since $\sum\limits_{l \in V_t^+} F^t_{Y_1^{t-1}l}=\frac{1}{N}$,
there exists $Y_j^{t}$ such that $F^t_{Y_1^{t-1}Y_j^{t}}>0$. Similarly, by using 
the second constraint we obtain that 
$F^t_{Y_{j'}^{t-1}Y_1^{t}}>0$. Let $\epsilon = 
\min(F^t_{Y_1^{t-1}Y_j^{t}},F^t_{Y_{j'}^{t-1}Y_1^{t}})$. Observe that if we 
increase $F^t_{Y_1^{t-1}Y_1^{t}}$, $F^t_{Y_{j'}^{t-1}Y_j^{t}}$ by $\epsilon$ and 
decrease $F^t_{Y_1^{t-1}Y_j^{t}}$, $F^t_{Y_{j'}^{t-1}Y_1^{t}}$ 
by $\epsilon$, we obtain another feasible solution. The cost difference of the 
two solutions is
$D=\epsilon (d(Y_1^{t-1},Y_j^{t}) + d(Y_{j'}^{t-1},Y_1^{t}) - d(Y_1^{t-1},Y_1^{t}) - 
d(Y_{j'}^{t-1},Y_j^{t}))$. If we prove that $D$ is
no negative, we are done. We show the latter using the fact that $Y_1^{t-1} \leq 
Y_{j'}^{t-1}$ and $Y_1^{t} \leq Y_j^{t}$. More precisely,
\begin{itemize}
 \item If $Y_1^{t-1} \leq Y_1^{t}$ then $D\geq 0$ since $Y_1^{t} \leq Y_j^{t}$.
 \item If $Y_1^{t-1} \geq Y_1^{t}$ then $D\geq 0$ since $Y_1^{t-1} \leq 
 Y_{j'}^{t-1}$.
\end{itemize}

Until now, we have shown that in the optimal solution, the node $Y_1^{t-1}$ 
sends all of her flow to the node $Y_1^{t}$. Meaning that
$Y_1^{t}$ does not receive flow by any other node apart from $Y_1^{t-1}$. By 
repeating the same argument,
it follows that in the optimal solution each node $Y_j^{t-1}$ sends all of her flow 
to $Y_j^{t}$.
\end{proof}

\tikzstyle{vertex}=[circle,fill=black!25,minimum size=27pt,inner sep=0pt]
\tikzstyle{weight} = [font=\small,above=0pt]
\tikzstyle{weight2} = [font=\small,above=4pt,left=9pt]
\tikzstyle{weight3} = [font=\small,below=5pt,left]
\tikzstyle{weight4} = [font=\small,below=10pt,left]
\tikzstyle{weight5} = [font=\small,below]
\tikzstyle{weight6} = [font=\small,above=12pt]

\tikzstyle{edge} = [draw,thick,->]
\usetikzlibrary{calc,arrows.meta,positioning}

\begin{figure}[ht] 
\centering	
\begin{tikzpicture}[scale=0.9]
\node[vertex] (1) at (0,0) {$Y_1^0$};
\node[vertex] (2) at (0,-1) {$Y_2^0$};
\node[vertex] (3) at (0,-5) {$Y_{K\cdot N}^0$};
\node at ($(2)!.5!(3)$) {\vdots};
\node[vertex] (4) at (3,0) {$Y_1^1$};
\node[vertex] (5) at (3,-1) {$Y_2^1$};
\node[vertex] (6) at (3,-5) {$Y_{K\cdot N}^1$};
\node at ($(5)!.5!(6)$) {\vdots};

\node[vertex] (10) at (6,0) {$Y_1^{t-1}$};
\node[vertex] (11) at (6,-2) {$Y_{j'}^{t-1}$};
\node[vertex] (12) at (6,-5) {$Y_{K\cdot N}^{t-1}$};
\node at ($(11)!.5!(12)$) {\vdots};
\node[vertex] (13) at (9,0) {$Y_1^t$};
\node[vertex] (14) at (9,-2.5) {$Y_j^t$};
\node[vertex] (15) at (9,-5) {$Y_{K\cdot N}^t$};
\node at ($(14)!.5!(15)$) {\vdots};

\node at ($(4)!.5!(10)$) {\ldots};
\node at ($(6)!.5!(12)$) {\ldots};
 \path[edge] (10) -- node[weight] {$F_{Y_1^{t-1}Y_1^t}^t$} (13);
 \path[edge] (10) -- node[weight6] {$F_{Y^{t-1}_1Y^t_j}^t$} (14);
 
  \path[edge] (11) -- node[weight4] {$F_{Y^{t-1}_{j'}Y^t_1}^t$} (13);
 \path[edge] (11) -- node[weight5] {$F_{Y^{t-1}_{j'}Y^t_j}^t$} (14);
 \path[edge] (2) -- node[weight] {} (5);
 \path[edge] (1) -- node[weight] {$F_{11}^1$} (4);
 \path[edge] (1) -- node[weight] {} (5);
 \path[edge] (2) -- node[weight2] {$F_{21}^1$} (4);
 \path[edge] (3) -- node[weight] {} (6);

 \node[vertex] (20) at (12,0) {$Y_1^{T-1}$};
\node[vertex] (21) at (12,-1) {$Y_2^{T-1}$};
\node[vertex] (22) at (12,-5) {$Y_{K\cdot N}^{T-1}$};
\node at ($(21)!.5!(22)$) {\vdots};
\node[vertex] (23) at (15,0) {$Y_1^T$};
\node[vertex] (24) at (15,-1) {$Y_2^T$};
\node[vertex] (25) at (15,-5) {$Y_{K\cdot N}^T$};
\node at ($(24)!.5!(25)$) {\vdots};

\node at ($(13)!.5!(20)$) {\ldots};
\node at ($(15)!.5!(22)$) {\ldots};

 \path[edge] (20) -- node[weight] {} (24);
 \path[edge] (20) -- node[weight] {$F_{11}^T$} (23);
 \path[edge] (21) -- node[weight] {} (24);
 \path[edge] (21) -- node[weight2] {$F_{21}^T$} (23);
 \path[edge] (22) -- node[weight] {} (25);
\node at ($(10)!.5!(11)$) {\vdots};
\node at ($(13)!.5!(14)$) {\vdots};
\path[edge] (12) -- node[weight] {} (15);

\end{tikzpicture}
\caption{The flow described by LP~(\ref{eq:LP2}).} \label{f:flow}
\end{figure}
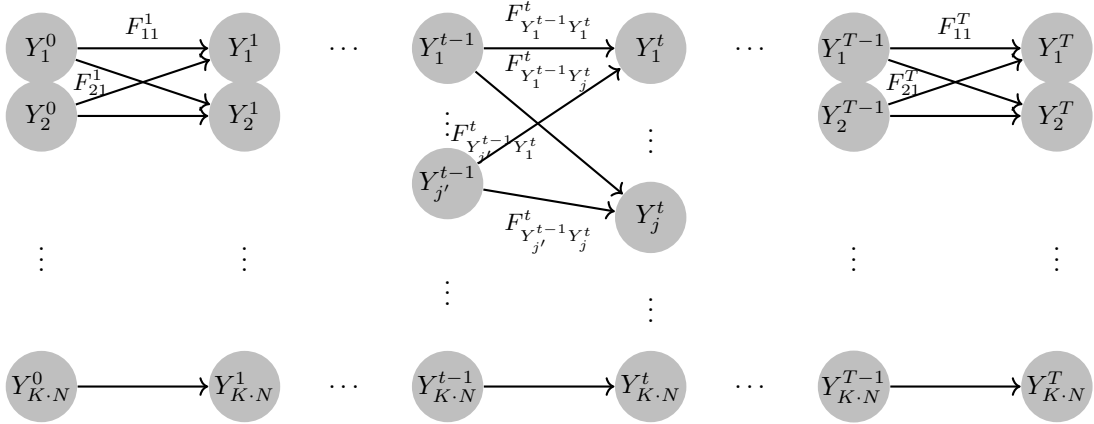

\vspace{0.5cm}
\noindent
\textbf{Lemma \ref{l:connection_cost}}
\emph{Let ConCost$^t_i(Sol_p)$ denote the connection cost of agent $i$ at 
stage $t$ in $Sol_p$ of Definition~\ref{d:Sol_p}. Then}

\[\frac{1}{N}\sum_{p=1}^N ConCost^t_i(Sol_p) =\sum_{i\in C}\sum_{j\in V} 
d(\text{Loc}(i,t),j)x_{ij}^t
\]

\begin{proof}
We will prove that $\frac{1}{N}\sum_{p=1}^N\text{ConCost}_i^t(Sol_p)$
equals $\sum_{j\in V} d(\text{Loc}(i,t),j)x_{ij}^t$.
We remind that by Assumption~\ref{a:1}, $c_j$ is $1/N$ if $j \in V_t^+$ and $0$ 
otherwise. As a result, in the optimal fractional solution, each agent $i$ finds the 
$N$ closest to $\text{Loc}(i,t)$ nodes of $V_t^+$ and receives a $1/N$ amount 
of service from each one of them.
Let us call this set $N_i^t$.
By Observation~\ref{o:1}, the
nodes in $N_i^t$ must be consecutive nodes of $V_t^+$ i.e. $N_i^t = 
\{Y_l^t,\dots,Y_{l+N-1}^t\}$ and 

\[\sum\limits_{j\in V} d(\text{Loc}(i,t),j)x_{ij}^t =  \sum\limits_{j = l}^{l+N-1} 
d(\text{Loc}(i,t),Y_j^t)/N\]

\noindent Since $Sol_p$ puts facilities in the positions $\{Y_{(m-1)\cdot N + 
p}^t\}_{m=1}^K$,
there exists a unique node $Y_{l(p)}^t \in N_i^t$ in which $Sol_p$ puts a facility.
$Y_{l(p)}^t$ is the closest node to $\text{Loc}(i,t)$ from all the nodes in which 
$Sol_p$ puts a facility.
As a result, $ConCost_i^t(Sol_p)= d(\text{Loc}(i),Y_{l(p)}^t)$. Now, summing  
over $p$ we get,
\begin{eqnarray*}
\frac{1}{N}\sum\limits_{p=1}^N ConCost_i^t(Sol_p) &=& 
\frac{1}{N}\sum\limits_{p=1}^N d(\text{Loc}(i),Y_{l(p)}^t)\\
&=& \sum\limits_{j=l}^{l+N-1} d(\text{Loc}(i),Y_j^t)/N\\
&=&\sum\limits_{j\in V} d(\text{Loc}(i,t),j)x_{ij}^t
\end{eqnarray*}
\end{proof}

\end{subsection}

\begin{subsection}{Omitted proofs of Section \ref{s:online}}\label{app:2}
\textbf{Lemma \ref{l:z1}}
\emph{Let $z = (z_1,z_2)$ denote the values of the variables $z_1,z_2$ after 
Step~1 of Algorithm~\ref{alg:double_coverage}. Then,}
\[\sum_{k=1}^2|z_k - x_k^{t-1}| \leq 2\sum_{k=1}^2 |y_k^t - y_k^{t-1}| -\Phi_t(z) 
+\Phi_{t-1}(x^{t-1})\]

\begin{proof}
Assume that $x_2^{t-1} \leq \alpha_1^t$, then 
Algorithm~\ref{alg:double_coverage} will  first
move facility $2$ to $\alpha_1^t$ ($z_1 = x_1^{t-1},z_2 = \alpha_1^t$), paying 
moving cost equal to $|\alpha_1^t - x_2^{t-1}|$.
This moving cost can be bounded with the use of the potential function $\Phi$. 
More specifically, we have that $\Phi_t(z)- \Phi_t(x^{t-1}) 
+ \Phi_t(x^{t-1}) - \Phi_{t-1}(x^{t-1})$ 
\begin{eqnarray}\label{in}
&=& \Phi_t(z) - \Phi_t(x^{t-1})\nonumber
+2 \sum_{k=1}^2(|y_k^t - x_k^{t-1}| - |y_k^{t-1} - x_k^{t-1}|)\nonumber\\
&\leq& \Phi_t(z) - \Phi_t(x^{t-1}) + 2\sum_{k=1}^2 |y_k^t - y_k^{t-1}|
\end{eqnarray}
In the considered case $z_1 = x_1^{t-1}, z_2 = \alpha_1^t$, the difference 
$\Phi_t(z) - \Phi_t(x^{t-1})$ in the potential function  equals the quantity $2(|y_2^t 
- \alpha_1^t| - |y_2^t - x_2^{t-1}|) + |x_1^{t-1}-\alpha_1^t| - |x_1^{t-1}-x_2^{t-1}|$. 
By the definition of solution $y$ in Lemma~\ref{l:competive_with_y},
either $y_1^t$ or $y_2^t$ lies in the interval $[\alpha_1^t,\alpha_n^t]$.  Since 
either $y_1^t$ or $y_2^t$ lies in the interval $[a_1^t,a_2^t]$ and $y_1^t \leq 
y_2^t$, we have that
$a_1^t \leq y_2^t$. Meaning that $z_2$ is closer to $y_2^t$ than $x_2^{t-1}$ 
and consequently $2(|y_2^t-a_1^t | - |y_2^t-x_2^{t-1}|) = -2|x_2^{t-1} - a_1^t|$. 
Therefore, $\Phi_t(z) - \Phi_t(x^{t-1})= -2|x_2^{t-1} - a_1^t|+ |x_1^{t-1}-\alpha_1^t| 
- |x_1^{t-1}-x_2^{t-1}|= - |a_1^t - x_2^{t-1}|= - |z_2 - x_2^{t-1}|$,
which completes the proof of Lemma~\ref{l:z1} for this case of Step~1.

Notice that inequality~(\ref{in}) holds for all three cases of Step~1. Thus, one just 
need to
prove that $\Phi_t(z) - \Phi_t(x^{t-1}) \leq - \sum_{k=1}^2|z_k - x_k^{t-1}|$
for the other two cases. We prove it for the third case of Step~1, since the 
second case ($x_1^{t-1} \geq \alpha_n^t$) is just symmetric to the first case.

In the third case of Step~1, we have that $x_1^{t-1}<a_1^t$, 
$x_2^{t-1}>a_n^t,z_1 = x_1^{t-1}+\min(|x_1^{t-1}-a_1^t|,|x_2^{t-1}-a_n^t|)$ and 
$z_2 = x_2^{t-1}-\min(|x_1^{t-1}-a_1^t|,|x_2^{t-1}-a_n^t|)$. The difference 
$\Phi_t(z) - \Phi_t(x^{t-1})$ in the potential function  equals the quantity $2(|z_1 - 
y_1^t| - |x_1^{t-1} - y_1^t|+|z_2-y_2^t|-|x_2^{t-1}-y_2^t|) + |z_1-z_2| - 
|x_1^{t-1}-x_2^{t-1}|$.
Now, $|z_1 - z_2|$ - $|x_1^{t-1} - x_2^{t-1}| = 
-2\min(|x_1^{t-1}-\alpha_1^t|,|x_2^{t-1}-\alpha_n^t|) = - \sum_{k=1}^2|z_k - 
x^{t-1}_k|$.
Assume that $y_1^t \in [a_1^t,a_n^t]$, then $\sum_{k=1}^2 (|z_k - y_k^t| - 
|x_k^{t-1} - y_k^t|)\leq 0$ since
$|z_1 - y_1^t| - |x_1^{t-1} - y_1^t| 
=-\min(|x_1^{t-1}-\alpha_1^t|,|x_2^{t-1}-\alpha_n^t|)$ and $|z_2 - y_2^t| - 
|x_2^{t-1} - y_2^t| \leq \min(|x_1^{t-1}-\alpha_1^t|,|x_2^{t-1}-\alpha_n^t|)$.
As a result, inequality~(\ref{in}) holds. Using the same argument in case $y_2^t 
\in [a_1^t,a_n^t]$ completes the proof.
\end{proof}

\vspace{0.5cm}

\noindent
\textbf{Lemma \ref{l:z2}}
\emph{Let $x^t = (x_1^t,x_2^t)$ denote the locations of facilities at stage $t$ 
after the execution of Step~2.
Then,} \[ \sum_{k=1}^2 [H(C_{kt}) + |x^t_k - z_k|]  \leq 21\sum_{k=1}^2H(C_{kt}^*) 
-\Phi_t(x^t) + \Phi_t(z)\]

\begin{proof}
Observe that by Algorithm~\ref{alg:double_coverage}, either $a_1^t \leq z_1 \leq 
a_n^t$
or $a_1^t \leq z_2 \leq a_n^t$. As a result, we need to prove the claim for the 
following 4 cases:

\begin{itemize}
 \item $a_1\leq z_1 \leq a_n$ \textbf{and} $z_2- a_n\geq 3H(C_t)$
 \item $a_1\leq z_1 \leq a_n$ \textbf{and} $z_2- a_n < 3H(C_t)$
 \item $a_1\leq z_2 \leq a_n$ \textbf{and} $a_1- z_1\geq 3H(C_t)$
 \item $a_1\leq z_2 \leq a_n$ \textbf{and} $a_1- z_1 < 3H(C_t)$
\end{itemize}
We will prove just the first and the second case since the third is symmetric to 
the first and the forth is symmetric to the second.

In case $a_1\leq z_1 \leq a_n$ and $z_2- a_n\geq 3H(C_t)$,  
Algorithm~\ref{alg:double_coverage} puts
facility $1$ in the median of $C_t$, namely $x_1^t =M_{C_t}$ (or $x_1^t \in 
M_{C_t}$ in case the number of agents is even), and moves facility $2$  to the 
left by a distance of $3H(C_t)$.

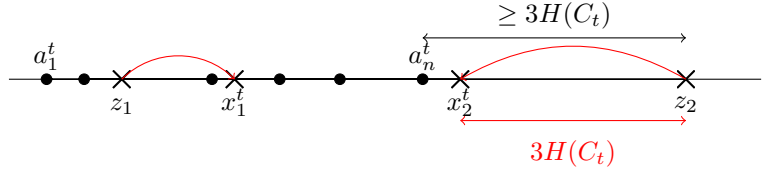
\begin{figure}[h]
\centering
\begin{tikzpicture}
\draw (0,0) -- (10,0);
\coordinate (y1p) at (1.5,0);
\coordinate (y1a) at (3,0);
\coordinate (y2p) at (9,0);
\coordinate (y2a) at  (6,0);
\coordinate (y1O) at  (2,0);
\coordinate (y2O) at  (6,0);
\coordinate (a1) at (0.5,0);
\coordinate (a2) at (5.5,0);
\filldraw (y1p) node[align=center,   below=4pt] {$z_1$} --
(y1p)  node[cross,thick,minimum size=4pt] {}     --
(y1a)  node[align=center, below] {$x_1^t$}     --
(y1a)  node[cross,thick,minimum size=4pt] {}     --
(y2p) node[align=center,  below=3pt] {$z_2$} --
(y2p)  node[cross,thick,minimum size=3pt] {}     --
(a1) circle (2pt) node[align=center, above] {$a_1^t$} --
(a2) circle (2pt) node[align=center, above] {$a_n^t$} --
(1,0) circle (2pt)  --
(2.7,0) circle (2pt)   --
(3.6,0) circle (2pt)  --
(4.4,0) circle (2pt)  --
(y2a)  node[cross,thick,minimum size=4pt] {}     --
(y2a) node[align=center, below] {$x_2^{t}$} ;
\draw[->,red] (y1p) to[out=45,in=135] (y1a);
\draw[->,red] (y2p) to[out=150,in=30]  (y2a);
\draw[<->] (9,0.55) -- node[above]{$\geq 3H(C_t)$} (5.5,0.55) ;
\draw[<->,red] (6,-0.55) -- node[below=4pt]{$3H(C_t)$} (9,-0.55) ;
\end{tikzpicture}
\end{figure}

First note that $\sum_{k=1}^2H(C_{kt}) \leq H(C_t)$ since $x_1^t\in M_{C_t}$. 
Then
$|x_1^t - z_1| \leq |a_1^t - a_n^t|\leq H(C_t)$ because both $x_1^t$ and $z_1$ 
lie in
the interval $[a_1^t,a_n^t]$ and $|x_2^t - z_2|=3H(C_t)$ by 
Algorithm~\ref{alg:double_coverage}. Therefore, we have that 
$\sum_{k=1}^2H(C_{kt}) + |x_k^t-z_k|\leq 5H(C_t)$. 

By the geometry of this case and the aforementioned bounds,
\begin{eqnarray}
 \Phi_t(x^t) - \Phi_t(z) &=& 2\sum_{k=1}^2\left(|x_k^t -y_k^t| - |z_k -y_k^t|\right)  + 
 |x_1^t - x_2^t| - |z_1-z_2| \nonumber\\
 &\leq& 2\sum_{k=1}^2\left(|x_k^t -y_k^t| - |z_k -y_k^t|\right)  -2H(C_t)\nonumber 
\end{eqnarray}

\noindent
Since $\sum_{k=1}^2\left[H(C_{kt}) + |x_k^t-z_k|\right] \leq 
5H(C_t)$, the 
challenge is bounding $\sum_{k=1}^2\left(|x_k^t -y_k^t| - |z_k -y_k^t|\right)$ by 
the connection cost
$\sum_{k=1}^2H(C_{kt}^*)$. In case $\sum_{k=1}^2H(C_{kt}^*) = H(C_t)$ 
meaning that $y^t$ connects all agents to just one facility things are quite easy, 
at least intuitively. 
The real difficulty arises when $C_{1t} \neq \emptyset$ and $C_{2t} \neq 
\emptyset$,
where $\sum_{k=1}^2H(C_{kt}^*)$ can be arbitrarily smaller than $H(C_t)$. As we 
will see in this case
$x^t$ gets closer to $y^t$ and the term $\sum_{k=1}^2\left(|x_k^t -y_k^t| - |z_k 
-y_k^t|\right)$
becomes negative.\\

We start with the most challenging case, where $C_{1t}^* \neq 
\emptyset$ and $C_{2t}^* \neq \emptyset$.
We remind that our goal is showing that $x^t$ gets closer to $y^t$. Since 
$C_{2t}^* \neq \emptyset$ and $y_2^t \in M_{C_{2t}^*}$ 
we get that $y_2^t \leq a_n^t$ and as a result $|x_2^t -y_2^t| - |z_2 -y_2^t| = 
|x_2^t -z_2| -3 H(C_t).$

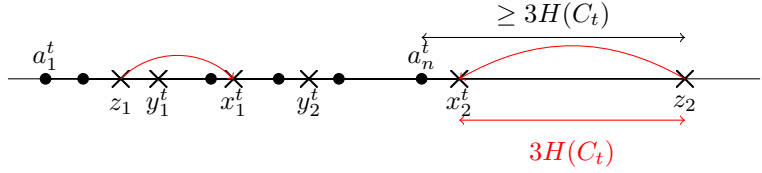
\begin{figure}[h]
	\centering
\begin{tikzpicture}\label{f:3}
\draw (0,0) -- (10,0);
\coordinate (y1p) at (1.5,0);
\coordinate (y1a) at (3,0);
\coordinate (y2p) at (9,0);
\coordinate (y2a) at  (6,0);
\coordinate (y1O) at  (2,0);
\coordinate (y2O) at  (4,0);
\coordinate (a1) at (0.5,0);
\coordinate (a2) at (5.5,0);
\filldraw (y1p) node[align=center,   below=4pt] {$z_1$} --
(y1p)  node[cross,thick,minimum size=4pt] {}     --
(y1O)  node[cross,thick,minimum size=4pt] {}     --
(y2O)  node[cross,thick,minimum size=4pt] {}     --
(y1a)  node[align=center, below] {$x_1^t$}     --
(y1a)  node[cross,thick,minimum size=4pt] {}     --
(y2p) node[align=center,  below=3pt] {$z_2$} --
(y2p)  node[cross,thick,minimum size=3pt] {}     --
(y1O)  node[align=center,  below] {$y_1^t$} --
(y2O)  node[align=center,  below] {$y_2^t$} --
(a1) circle (2pt) node[align=center, above] {$a_1^t$} --
(a2) circle (2pt) node[align=center, above] {$a_n^t$} --
(1,0) circle (2pt)  --
(2.7,0) circle (2pt)   --
(3.6,0) circle (2pt)  --
(4.4,0) circle (2pt)  --
(y2a)  node[cross,thick,minimum size=4pt] {}     --
(y2a) node[align=center, below] {$x_2^{t}$} ;
\draw[->,red] (y1p) to[out=45,in=135] (y1a);
\draw[->,red] (y2p) to[out=150,in=30]  (y2a);
\draw[<->] (9,0.55) -- node[above]{$\geq 3H(C_t)$} (5.5,0.55) ;
\draw[<->,red] (6,-0.55) -- node[below=4pt]{$3H(C_t)$} (9,-0.55) ;
\end{tikzpicture}
\end{figure}

\begin{eqnarray}
\Phi_t(x^t) - \Phi_t(z) &\leq& 2\sum_{k=1}^2\left(|x_k^t -y_k^t| - |z_k -y_k^t|\right) 
-2H(C_t) \nonumber\\
&=& 2\left(|x_1^t -y_1^t| - |z_1 -y_1^t|\right) + 2\left(|x_2^t -z_2| - |z_2 
-y_2^t|\right) -2H(C_t) \nonumber\\
&\leq&  2|x_1^t -z_1| -8H(C_t)\nonumber\\
&\leq&  2H(C_t) -8H(C_t)\nonumber\\
&\leq&  -6H(C_t)\nonumber\\
&\leq& \sum_{k=1}^2H(C_{kt}^*) -\sum_{k=1}^2[ H(C_{kt}) + |x_k^t 
-z_k|]\nonumber
\end{eqnarray}

\noindent Now, assume that $C_{1t}^* = \emptyset$ or $C_{2t}^* = \emptyset$ 
meaning that $\sum_{k=1}^2 H(C_{kt}^*) = H(C_t)$.
As a result, bounding \emph{everything} by $H(C_t)$ serves our purpose. More 
formally,
\begin{eqnarray*}
\Phi_t(x^t) - \Phi_t(z) &\leq& 2\sum_{k=1}^2\left(|x_k^t- y_k^t| - |z_k -y_k^t| \right) 
-2H(C_t)\\
&\leq& 2\sum_{k=1}^2|x_k^t -z_k| -2H(C_t)\\
&\leq& 6H(C_t)\\
&\leq& 11 H(C_t) -\sum_{k=1}^2 \left[ H(C_{kt}) + |x_k^t -z_k| \right]\\
&=& 11\sum_{k=1}^2H(C_{kt}^*) -\sum_{k=1}^2\left[ H(C_{kt}) + |x_k^t -z_k| \right]
\end{eqnarray*}

\noindent The forth inequality follows from the fact that $\sum_{k=1}^2H(C_{kt}) 
+ |x_k^t-z_k|\leq 5H(C_t)$.

We now need to treat the second case where $a_1\leq z_1 \leq a_n$ and $z_2- 
a_n < 3H(C_t)$. Since Algorithm~\ref{alg:double_coverage}
computes the optimal clustering $(C_{1t},C_{2t})$ and puts $x_1^t$ in the interval 
$M_{C_{1t}}$ and 
$x_2^t$ in the interval $M_{C_{2t}}$, we are ensured that the connection cost of 
our solution is less than
the connection cost of $y^t$, $\sum_{k=1}^2H(C_{kt}) \leq 
\sum_{k=1}^2H(C_{kt}^*)$, so we are mostly concerned in bounding
$\sum_{k=1}^2|x_k^t -z_k|$.

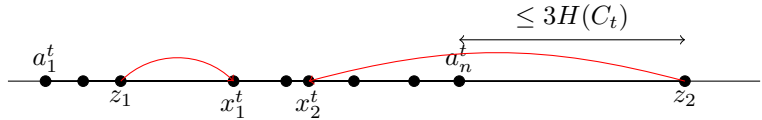
\begin{figure}[h]
\centering
\begin{tikzpicture}
\draw (0,0) -- (10,0);
\coordinate (y1p) at (1.5,0);
\coordinate (y1a) at (3,0);
\coordinate (y2p) at (9,0);
\coordinate (y2a) at  (4,0);
\coordinate (y1O) at  (2,0);
\coordinate (y2O) at  (6,0);
\coordinate (a1) at (0.5,0);
\coordinate (a2) at (6,0);
\filldraw (y1p) circle (2pt) node[align=center,   below] {$z_1 $} --
(y1a) circle (2pt) node[align=center, below] {$x_1^t$}     --
(y2p) circle (2pt) node[align=center,  below] {$z_2$} --
(a1) circle (2pt) node[align=center, above] {$a_1^t$} --
(a2) circle (2pt) node[align=center, above] {$a_n^t$} --
(1,0) circle (2pt)  --
(3.7,0) circle (2pt)   --
(4.6,0) circle (2pt)  --
(5.4,0) circle (2pt)  --
(y2a) circle (2pt) node[align=center, below] {$x_2^{t}$} ;
\draw[->,red] (y1p) to[out=45,in=135] (y1a);
\draw[->,red] (y2p) to[out=165,in=15]  (y2a);
\draw[<->] (9,0.55) -- node[above]{$\leq 3H(C_t)$} (6,0.55) ;
\end{tikzpicture}
\end{figure}

\noindent The easy case is when $\sum_{k=1}^2H(C_{kt}^*) = H(C_t)$. A small 
difference with the previous case is that we don't know how $|x_2^t - z_2|$ is. 
However,  $z_1, x_1^t, x_2^t \in [a_1^t, \ldots a_n^t] $ and 
$|x_2^t-z_2|=|x_2^t-a_n^t|+|a_n^t-z_2|$. Thus, $|x_1^t-z_1|+|x_2^t-a_n^t|\leq 
H(C_t)$, $|a_n^t-z_2| \leq 3H(C_t)$ and therefore $\sum_{k=1}^2 [H(C_{kt}) + 
|x_k^t-z_k|]\leq 5H(C_t)$. So we can again bound \emph{everything}
by $H(C_t).$

\begin{eqnarray*}
\Phi_t(x^t) - \Phi_t(z) &=& 2\sum_{k=1}^2\left(|x_k^t -y_k^t| - |z_k -y_k^t|\right) + 
|x_1^t - x_2^t| - |z_1 - z_2|\\
&\leq&  3\sum_{k=1}^2|x_k^t -z_k|\\
&\leq& 4\sum_{k=1}^2|x_k^t -z_k| - \sum_{k=1}^2|x_k^t -z_k|\\
&\leq& 20H(C_t) -\sum_{k=1}^2\left[|x_k^t -z_k|\right]\\
&\leq& 21\sum_{k=1}^2H(C_{kt}^*) -\sum_{k=1}^2\left[ H(C_{kt}) + |x_k^t 
-z_k|\right]
\end{eqnarray*}
\noindent Things become more complicated, when the connection cost 
$\sum_{k=1}^2H(C_{kt}^*)$ is
relatively small ($C_{1t}^* \neq \emptyset$ and $C_{2t}^* \neq \emptyset$), 
where bounding everything by $H(C_t)$
does not work. However, the solutions $x^t$ and $y^t$ will be relatively close in 
this case. More formally,
\begin{eqnarray*}
\Phi_t(x^t) - \Phi_t(z) &=& 2\sum_{k=1}^2\left(|x_k^t -y_k^t| - |z_k -y_k^t|\right) + 
|x_1^t - x_2^t| - |z_1 - z_2|\\
&=& 2\sum_{k=1}^2|x_k^t -y_k^t| - 2\sum_{k=1}^2|z_k -y_k^t| + 
\sum_{k=1}^2|x_k^t - z_k|\\
&=& 2\sum_{k=1}^2|x_k^t -y_k^t| + 2\sum_{k=1}^2\left(|x_k^t-z_k|-|z_k 
-y_k^t|\right) - \sum_{k=1}^2|x_k^t-z_k|\\
&\leq& 4\sum_{k=1}^2|x_k^t -y_k^t| - \sum_{k=1}^2|x_k^t-z_k|\\
\end{eqnarray*}
We need to upper bound the distance $\sum_{k=1}^2|x_k^t -y_k^t|$. Observe 
that in the solution $x^t$,
the agent at position $a_1^t$ connects to the left facility (facility 1) and  the 
agent at position $a_n^t$ 
connects to the right facility (facility 2), $|x_1^t - a_1^t| + |x_2^t - a_n^t| \leq 
\sum_{k=1}^2H(C_{kt}).$ 
Since $C_{1t}^*\neq \emptyset$ and $C_{2t}^*\neq \emptyset$, the same holds 
for the solution $y^t$. As a result, 
\begin{eqnarray*}
\Phi_t(x^t) - \Phi_t(z) &\leq& 4\sum_{k=1}^2|x_k^t -y_k^t| - 
\sum_{k=1}^2|x_k^t-z_k|\\
&\leq& 4\left(|x_1^t -a_1^t|+|y_1^t -a_1^t| + |x_2^t -a_n^t| + |y_2^t -a_n^t|\right) 
- \sum_{k=1}^2|x_k^t-z_k|\\
&\leq& 4\sum_{k=1}^2\left[H(C_{kt}) + H(C_{kt}^*)\right] - 
\sum_{k=1}^2|x_k^t-z_k|\\
&\leq& 9\sum_{k=1}^2H(C_{kt}^*) - \sum_{k=1}^2\left[ H(C_{kt}) + |x_k^t-z_k|\right]
\end{eqnarray*}
\end{proof}

\end{subsection}

\end{document}